\documentclass{article}
\usepackage{amssymb, amsmath}
\usepackage[english]{babel}
\usepackage{amsthm}
\usepackage{commath}
\usepackage{graphicx}
\usepackage{algorithm}
\usepackage{algorithmic}
\usepackage{amsfonts}
\usepackage{mathtools}
\usepackage{verbatim}

\newcommand\myeq{\stackrel{\mathclap{\normalfont\mbox{def}}}{=}}
\newtheorem{theorem}{Theorem}

\author{
  Changye Wu\thanks{CEREMADE, Universit\'e Paris-Dauphine PSL, France. wu@ceremade.dauphine.fr}
  \and
  Christian Robert\thanks{Universit\'e Paris Dauphine PSL, CREST, France and University of Warwick, UK.  xian@ceremade.dauphine.fr}
}
\title{Generalized Bouncy Particle Sampler}
\begin{document}
\maketitle
\begin{abstract}
\noindent As a special example of piecewise deterministic Markov process, bouncy particle sampler is a rejection-free, irreversible Markov chain Monte Carlo algorithm and can draw samples from target distribution efficiently. We generalize bouncy particle sampler in terms of its transition dynamics. In BPS, the transition dynamic at event time is deterministic, but in GBPS, it is random. With the help of this randomness, GBPS can overcome the reducibility problem in BPS without refreshment.
\end{abstract}
\section{Introduction}
As a powerful sampling technique, Markov chain Monte Carlo (MCMC) method has been widely used in computational statistics and is now a standard tool in Bayesian inference, where posterior distribution is often intractable analytically, known up to a constant. However, almost all existing MCMC algorithms, such as Metropolis-Hastings algorithm (MH), Hamiltonian Monte Carlo (HMC) and Metropolis adjusted Langevin algorithm (MALA), are based on detailed balance condition, dating back to (\cite{m1953}, \cite{hastings1970monte}). Recently, a novel type of MCMC method --- piecewise deterministic Markov process (PDMP) --- appeared in computational statistics, method that is generally irreversible, meaning a violation of the detailed balance condition. The theory of PDMP is developed by (\cite{davis1984piecewise}, \cite{davis1993markov}), while its quite remarkable applications on computational statistics are implemented by (\cite{bouchard2017bouncy}, \cite{bierkens2015piecewise}, \cite{bierkens2016zig}). \\
\\
Compared with traditional MCMC algorithms, PDMP is rejection-free, which means that there is no waste of proposal samples. Based on detailed balance condition, traditional MCMC algorithms are reversible. However, some theoretic work and  numerical experiments (\cite{hwang1993accelerating}, \cite{sun2010improving}, \cite{chen2013accelerating}, \cite{bierkens2016non}) have shown that irreversible Markov chain can outperform reversible MCMC with respect to mixing rate and asymptotic variance. PDMP is a typically irreversible Markov chain, which is of interest to investigate. Bouncy particle sampler (BPS) generates a special piecewise deterministic Markov chain, originating in \cite{peters2012rejection} and being explored by  \cite{bouchard2017bouncy}. Zig-zag process sampler \cite{bierkens2016zig} is another PDMP example and \cite{fearnhead2016piecewise} unifies the BPS and zig-zag process sampler in the framework of PDMP. Besides, MCMC algorithms are difficult to scale, since computing each MH acceptance ratio needs to sweep over the whole data set. However, according to their special structures, BPS and zig-zag process sampler are easy to scale for big data. Except PDMP style MCMC algorithms, almost all other existing scalable MCMC algorithms (such as, \cite{scott2016bayes}, \cite{neiswanger2013asymptotically}, \cite{wang2013parallelizing}, \cite{minsker2014robust}, \cite{quiroz2015speeding}, \cite{bardenetmarkov}), are approximate, not exact, which do not admit the target distribution as their invariant distribution. \\
\\
In this article, we generalize bouncy particle sampler -- generalized bouncy particle sampler (GBPS) -- which can be treated as an extension of BPS and zig-zag process sampler. In BPS, the transition dynamic at event time is deterministic. In opposition, the transition dynamic in GBPS is random with respect to some distribution. In zig-zag process sampler, we decompose the velocity with respect to some fixed coordinate system, while in GBPS, we use a moving coordinate system to decompose the velocity. Besides, the main gain of GBPS compared with BPS is that there is no parameter to tune. In fact, for BPS, in order to overcome the reducibility problem, we need to add a refreshment Poisson process to refresh the velocity occasionally, which needs to be tuned to balance the efficiency and accuracy. But for GBPS, the randomness of refreshment is incorporated into the transition dynamics and there is no parameter to tune.
\\
\\
This paper is organized as follows: we introduce piecewise deterministic Markov process and bouncy particle sampler in section 2, followed by the introduction of  generalized bouncy particle sampler (GBPS) and its implementation issues in section 3. In section 4, we present three numerical experiments of GBPS. At last, we discuss some questions and conclude in section 5.
\section{Piecewise Deterministic Markov Process}
\noindent In this section, suppose $\pi(\bold{x})$ be the target distribution, where $\bold{x}\in\mathbb{R}^d$. We introduce an auxiliary variable, $\bold{v}\in\mathbb{R}^d$, called velocity, which is restricted to be of unit length, $\Vert\bold{v}\Vert_2=1$. Denote $\bold{z} = (\bold{x}, \bold{v})\in\mathbb{R}^d\times S_{d-1}$. In order to obtain the target, we just need to force $\pi(\bold{x})$ to be the marginal distribution of $\pi(\bold{x}, \bold{v})$ with respect to $\bold{x}$. Let $\{\bold{z}_t\}$ denote a piecewise deterministic Markov chain of $\bold{z}$ on the augmented space $(\bold{x},\bold{v})\in\mathbb{R}^d\times S_{d-1}$. The dynamics of PDMP consist of three types of dynamics, namely, deterministic dynamic, event occurrence and transition dynamic. Specifically, 
\begin{enumerate}
  \item \textbf{The deterministic dynamic}: between two event times, the Markov process evolves deterministically, according to some partial differential equation:
           \begin{equation*}
           \frac{dz^{(i)}_t}{dt} = \Psi^{(i)}(\bold{z}_t),\quad i = 1, \cdots, 2d
           \end{equation*}
  \item \textbf{The event occurrence}: the event occurs at the rate: $\lambda(\bold{z}_t)$.
  \item \textbf{The transition dynamic}: At the event time, $\tau$, we denote $\bold{z}_{\tau-}$ the state prior to $\tau$, then $\bold{z}_{\tau}\sim Q(\cdot|\bold{z}_{\tau-})$
\end{enumerate}
Following from (\cite{davis1993markov}, Theorem 26.14), this Markov process's extension generator is 
\begin{equation*}
\mathcal{A} f(\bold{z}) = \nabla f(\bold{z})\cdot\Psi(\bold{z}) + \lambda(\bold{z})\int_{\mathbb{R}^{d}\times S_{d-1}}\left(f(\bold{z}')-f(\bold{z})\right)Q(d\bold{z}';\bold{z})
\end{equation*}
\subsection{Bouncy Particle Sampler}
Bouncy particle sampler (BPS) is a specific piecewise deterministic Markov process, which admits $\pi(\bold{x})d\bold{x}\otimes d\bold{v}$ over the state space $\mathbb{R}^d\times S_{d-1}$ as its invariant distribution, by specifying the event rate $\lambda(\bold{z})$ and the transition dynamic $Q(d\bold{z}'; \bold{z})$.
\begin{enumerate}
  \item \textbf{The deterministic dynamic}: 
           \begin{equation*}
           \frac{dx^{(i)}_t}{dt} = v^{(i)}_t, \quad \frac{dv^{(i)}_t}{dt} = 0, \quad i = 1, \cdots, d
           \end{equation*}
  \item \textbf{The event occurrence}: $\lambda(\bold{z}_t) = \max\{0, -\bold{v}_t\cdot\nabla\log\pi(\bold{x}_t)\}$.
  \item \textbf{The transition dynamic}: $Q(\cdot|\bold{x}, \bold{v}) = \delta_{(\bold{x}, P_{\bold{x}}\bold{v})}(\cdot)$, where 
           \begin{equation*}
           P_{\bold{x}}\bold{v} = \bold{v} -2 \frac{\langle\bold{v}, \nabla\log\pi(\bold{x})\rangle}{\langle\nabla\log\pi(\bold{x}), \nabla\log\pi(\bold{x})\rangle}\nabla\log\pi(\bold{x})
           \end{equation*}
\end{enumerate}
\cite{bouchard2017bouncy} has shown that BPS admits  $\pi(\bold{x})d\bold{x}\otimes d\bold{v}$ as its invariant distribution. However, the authors also find that pure BPS (specified above) meets with a reducibility problem and add a reference Poisson process into BPS to overcome it. The workflow of BPS with refreshment is shown in Algorithm \ref{alg:BPS}.
\begin{algorithm}[h]
   \caption{Bouncy Particle Sampler}
   \label{alg:BPS}
\begin{algorithmic}
\STATE {\bfseries Initialize:} $\bold{x}_0, \bold{v}_0, T_0 = 0$.
\FOR{ $i = 1, 2, 3, \cdots$}
    \STATE Generate $\tau \sim PP(\lambda(\bold{x}_t, \bold{v}_t))$
    \STATE Generate $\tau^{\text{ref}} \sim PP(\lambda^{\text{ref}})$
    \IF {$\tau \leq \tau^{\text{ref}}$}
        \STATE $T_i \leftarrow T_{i-1} + \tau$
        \STATE $\bold{x}_i \leftarrow \bold{x}_{i-1} + \tau\bold{v}_{i-1}$
        \STATE $\bold{v}_i \leftarrow \bold{v}_{i-1} - 2 \frac{\langle\bold{v}_{i-1}, \nabla\log\pi(\bold{x}_i)\rangle}{\langle\nabla\log\pi(\bold{x}_{i}), \nabla\log\pi(\bold{x}_i)\rangle}\nabla\log\pi(\bold{x}_i)$
    \ELSE
        \STATE $T_i \leftarrow T_{i-1} + \tau^{\text{ref}}$
        \STATE $\bold{x}_i \leftarrow \bold{x}_{i-1} + \tau^{\text{ref}}\bold{v}_{i-1}$
        \STATE $\bold{v}_i \sim \mathcal{U}(S_{d-1})$
    \ENDIF
    \ENDFOR
   \end{algorithmic}
\end{algorithm}
\section{Generalized Bouncy Particle Sampler}
\noindent In BPS, at event time, the velocity changes deterministically. However, we find that the velocity can be changed into other directions, according to some distribution, at event time, which incorporates the randomness of the reference Poisson process in BPS to overcome the reducibility. In  this section, we generalize the BPS. Specifically, prior to event time, we decompose the velocity according to the gradient of $\log{\pi(\bold{x})}$, flip the parallel subvector and resample the orthogonal subvector with respect to some distribution. The details are as follows: 
\begin{enumerate}
  \item \textbf{The deterministic dynamic}: 
           \begin{equation*}
           \frac{dx^{(i)}_t}{dt} = v^{(i)}_t, \quad \frac{dv^{(i)}_t}{dt} = 0, \quad i = 1, \cdots, d
           \end{equation*}
  \item \textbf{The event occurrence}: $\lambda(\bold{z}_t) =\max\{0,-\langle\bold{v}_t,\nabla\log\pi(\bold{x}_t)\rangle\}$.
  \item \textbf{The transition dynamic}: $Q(d\bold{x}',d\bold{v}'|\bold{x}, \bold{v}) = \delta_{\{\bold{x}\}}(d\bold{x}')\delta_{\{-\bold{v}_1\}}(d\bold{v}'_1)\mathcal{N}_{\bold{v}_1^{\perp}}(d\bold{v}'_2)$, where
           \begin{equation*}
           \bold{v}_1 = \frac{\langle \bold{v}, \nabla\log\pi(\bold{x})\rangle}{\langle\nabla\log\pi(\bold{x}), \nabla\log\pi(\bold{x})\rangle}\nabla\log\pi(\bold{x}), \quad \bold{v}_2 = \bold{v} - \bold{v}_1
           \end{equation*}
           \begin{equation*}
           \bold{v}'_1 = \frac{\langle \bold{v}', \nabla\log\pi(\bold{x})\rangle}{\langle\nabla\log\pi(\bold{x}), \nabla\log\pi(\bold{x})\rangle}\nabla\log\pi(\bold{x}),\quad \bold{v}'_2 = \bold{v}' - \bold{v}'_1
           \end{equation*}
           \begin{equation*}
           \bold{v}_1^{\perp} = \left\{\bold{u}\in\mathbb{R}^d : \langle\bold{u},\bold{v}_1\rangle = 0\right\}
           \end{equation*}
           $\mathcal{N}_{\bold{v}_1^{\perp}}$ is the $(d-1)-$dimensional standard normal distribution over the space $\bold{v}_1^{\perp}$.
\end{enumerate}
We summarize the GBPS in Algorithm \ref{alg:GBPS}.
\begin{algorithm}[h]
   \caption{Generalized Bouncy Particle Sampler}
   \label{alg:GBPS}
\begin{algorithmic}
\STATE {\bfseries Initialize:} $\bold{x}_0, \bold{v}_0, T_0 = 0$.
\FOR{ $i = 1, 2, 3, \cdots$}
    \STATE Generate $\tau \sim PP(\lambda(\bold{x}_t, \bold{v}_t))$
    \STATE $T_i \leftarrow T_{i-1} + \tau$
    \STATE $\bold{x}_i \leftarrow \bold{x}_{i-1} + \tau\bold{v}_{i-1}$
    \STATE $\bold{v}_i \leftarrow Q(d\bold{v}|\bold{x}_i, \bold{v}_{i-1})$
    \ENDFOR
   \end{algorithmic}
\end{algorithm}
\begin{theorem}
The above piecewise deterministic Markov chain admits $\pi(\bold{x})d\bold{x}\otimes \psi_d(\bold{v})d\bold{v}$ over $\mathbb{R}^{2d}$ as its invariant distribution, where $\psi_d(\bold{v})$ is the density function of $d-$dimensional standard normal distribution.
\end{theorem}
\begin{proof}
In order to prove $\pi(\bold{x})d\bold{x}\otimes\psi_d(\bold{v})d\bold{v}$ is the invariant distribution of generator $\mathcal{A}$ of the above Markov chain, we just need to prove the following equation is satisfied by appropriate functions $f$:
\begin{equation*}
\int_{\mathbb{R}^{d}}\int_{\mathbb{R}^d}\mathcal{A}f(\bold{x},\bold{v})\pi(\bold{x})\psi_d(\bold{v})d\bold{x}d\bold{v} = 0
\end{equation*}
where by Theorem 26.14, \cite{davis1993markov}
\begin{equation*}
\mathcal{A} f(\bold{z}) = \nabla f(\bold{z})\cdot\Psi(\bold{z}) + \lambda(\bold{x},\bold{v})\int_{\bold{v}'\in\mathbb{R}^d}f(\bold{x},\bold{v}')Q(d\bold{v}'|\bold{x},\bold{v}) - \lambda(\bold{x},\bold{v})f(\bold{x},\bold{v})
\end{equation*}
Since for bounded $f$,
\begin{equation*}
\begin{split}
&\int_{\mathbb{R}^{d}}\int_{\mathbb{R}^d}\langle\nabla f(\bold{z}), \Psi(\bold{z})\rangle\pi(\bold{x})\psi_d(\bold{v})d\bold{x}d\bold{v}\\
&= \int_{\mathbb{R}^{d}}\int_{\mathbb{R}^d}\langle\nabla_{\bold{x}} f(\bold{x},\bold{v}), \bold{v}\rangle\pi(\bold{x})\psi_d(\bold{v})d\bold{x}d\bold{v} \\
&= \int_{\mathbb{R}^{d}}\int_{\mathbb{R}^d}\langle-\nabla\log\pi(\bold{x}), \bold{v}\rangle f(\bold{x},\bold{v})\pi(\bold{x})\psi_d(\bold{v})d\bold{x}d\bold{v}
\end{split}
\end{equation*}
For each $\bold{x}$, decomposing the velocity spaces, $\mathbb{R}^d$, into the direct sum, $\bold{v}_1^{\parallel}\oplus\bold{v}_1^{\perp}$, such that $\bold{v}_1 = (\Vert\bold{v}_1\Vert_2,0,\cdots,0)$, then
\begin{equation*}
\lambda(\bold{x},(v_1,\cdots,v_d)) = \max\left\{0,-\langle(v_1,0,\cdots,0),\nabla\log\pi(\bold{x})\rangle\right\}
\end{equation*}
\begin{equation*}
\begin{split}
&\int_{\bold{v}\in\mathbb{R}^d}\int_{\bold{v}'\in\mathbb{R}^d}f(\bold{x},\bold{v}')\lambda(\bold{x},\bold{v})\pi(\bold{x})\psi_d(\bold{v})Q(d\bold{v}'|\bold{x},\bold{v})d\bold{v}\\
&=\int_{v_1\in\mathbb{R}}\int_{(v_2,\cdots,v_d)\in\mathbb{R}^{d-1}}\int_{v'_1\in\mathbb{R}}\int_{(v'_2,\cdots,v'_d)\in\mathbb{R}^{d-1}}f(\bold{x},(v_1',\cdots,v'_d))\\
&\times\lambda(\bold{x},(v_1,\cdots,v_d))\pi(\bold{x})\psi_1(v_1)\psi_{d-1}(v_2,\cdots,v_d)\psi_{d-1}(v'_2,\cdots,v'_d)\\
&\times\delta_{\{-v_1\}}(v'_1)dv_1dv_2\cdots dv_ddv'_1dv'_2\cdots dv'_d\\
&=\int_{v_1\in\mathbb{R}}\int_{(v_2,\cdots,v_d)\in\mathbb{R}^{d-1}}\int_{(v'_2,\cdots,v'_d)\in\mathbb{R}^{d-1}}f(\bold{x},(-v_1,v'_2,\cdots,v'_d))\\
&\times\max\left\{0,-\langle(v_1,0,\cdots,0),\nabla\log\pi(\bold{x})\rangle\right\}\pi(\bold{x})\psi_1(v_1)\psi_{d-1}(v_2,\cdots,v_d)\\
&\times\psi_{d-1}(v'_2,\cdots,v'_d) dv_1dv_2\cdots dv_ddv'_2\cdots dv'_d\\
&\quad\quad\text{(by the change of variable: $v_1 \rightarrow -v_1$ and $\psi_1(-v_1) = \psi_1(v_1)$)}\\
&=\int_{v_1\in\mathbb{R}}\int_{(v_2,\cdots,v_d)\in\mathbb{R}^{d-1}}\int_{(v'_2,\cdots,v'_d)\in\mathbb{R}^{d-1}}f(\bold{x},(v_1,v'_2,\cdots,v'_d))\\
&\times\max\left\{0,-\langle(-v_1,0,\cdots,0),\nabla\log\pi(\bold{x})\rangle\right\}\pi(\bold{x})\psi_1(-v_1)\psi_{d-1}(v_2,\cdots,v_d)\\
&\times \psi_{d-1}(v'_2,\cdots,v'_d)dv_1dv_2\cdots dv_ddv'_2\cdots dv'_d\\
&=\int_{v_1\in\mathbb{R}}\int_{(v_2,\cdots,v_d)\in\mathbb{R}^{d-1}}\int_{(v'_2,\cdots,v'_d)\in\mathbb{R}^{d-1}}f(\bold{x},(v_1,v'_2,\cdots,v'_d))\\
&\times\max\{0,\langle(v_1,0,\cdots,0),\nabla\log\pi(\bold{x})\rangle\}\pi(\bold{x})\psi_1(v_1)\psi_{d-1}(v_2,\cdots,v_d)\\
&\times \psi_{d-1}(v'_2,\cdots,v'_d)dv_1dv_2\cdots dv_ddv'_2\cdots dv'_d\\
&=\int_{v_1\in\mathbb{R}}\int_{(v'_2,\cdots,v'_d)\in\mathbb{R}^{d-1}}f(\bold{x},(v_1,v'_2,\cdots,v'_d))\psi_1(v_1)\psi_{d-1}(v'_2,\cdots,v'_d)\\
&\times\max\{0,\langle(v_1,0,\cdots,0),\nabla\log\pi(\bold{x})\rangle\}\pi(\bold{x})dv_1dv'_2\cdots dv'_d\\
&=\int_{\bold{v}\in\mathbb{R}^d}f(\bold{x},\bold{v})\lambda(\bold{x},-\bold{v})\pi(\bold{x})\psi_d(\bold{v})d\bold{v}
\end{split}
\end{equation*}
As a result, 
\begin{equation*}
\begin{split}
&\int_{\bold{x}\in\mathbb{R}^d}\int_{\bold{v}\in\mathbb{R}^d}\mathcal{A}f(\bold{x},\bold{v})\pi(\bold{x})\psi_d(\bold{v})d\bold{x}d\bold{v}\\
&=\int_{\bold{x}\in\mathbb{R}^d}\int_{\bold{v}\in\mathbb{R}^d}\left[ \langle-\nabla\log\pi(\bold{x}), \bold{v}\rangle \right]f(\bold{x},\bold{v})\pi(\bold{x})\psi_d(\bold{v})d\bold{x}d\bold{v}\\
&+\int_{\bold{x}\in\mathbb{R}^d}\int_{\bold{v}\in\mathbb{R}^d}\left[ \lambda(\bold{x}, -\bold{v})- \lambda(\bold{x}, \bold{v})  \right]f(\bold{x},\bold{v})\pi(\bold{x})\psi_d(\bold{v})d\bold{x}d\bold{v}\\
&\quad\quad\text{(since $\lambda(\bold{x},-\bold{v}) - \lambda(\bold{x},\bold{v}) = \langle\bold{v}, \nabla\log\pi(\bold{x})\rangle$)}\\
&=\int_{\mathbb{R}^d}\int_{\mathbb{R}^d}\left[ \langle-\nabla\log\pi(\bold{x}), \bold{v}\rangle + \langle\bold{v},\nabla\log\pi(\bold{x})\rangle  \right]f(\bold{x},\bold{v})\pi(\bold{x})d\bold{x}d\bold{v}\\
&=0
\end{split}
\end{equation*}
\end{proof}
In order to establish the ergodicity theorem of GBPS, we propose an assumption on the target distribution $\pi(\bold{x})$. \\
\\
\noindent \textbf{Assumption 1}: For any two points $\bold{x}_1, \bold{x}_2 \in \mathbb{R}^d$ and any velocity $\bold{v}\in\mathbb{R}^d, \Vert \bold{v} \Vert_2 = 1$, there exists $t > 0$, such that 
\begin{equation*}
\bold{x}_2\in S^{\perp}(\bold{x}_1 + t\bold{v}, \bold{v})
\end{equation*}
\begin{theorem}
Under Assumption 1, the Markov chain $\bold{z}'_t = (\bold{x}_t, \frac{\bold{v}_t}{\Vert \bold{v}_t\Vert})$ induced by GBPS admits $\pi(\bold{x})\times\mathcal{U}(S_{d-1})$ as its unique invariant distribution.
\end{theorem}
\noindent The proof of Theorem 2 and the definitions of notations in Assumption 1 can be found in Appendix.  Whether Theorem 2 is still correct without Assumption 1 is an open question.
\subsection{Construction of Estimator}
\noindent While constructing an unbiased estimator of $I = \int h(\bold{x})\pi(d\bold{x})$, we cannot use the  skeleton of the simulated GBPS path directly. In fact, such an estimator is biased. Suppose $\{\bold{x}_i, \bold{v}_i, T_i\}_{i=0}^M$ be the skeleton of an simulated trajectory, which means that at event time $T_i$, the state is $(\bold{x}_i, \bold{v}_i)$. Then, the whole trajectory $\bold{x}_{[0, T_M]}$ is filled up with
\begin{equation*}
\bold{x}_t = \bold{x}_i + (t - T_i)\bold{v}_i, \quad T_i \leq t < T_{i+1}
\end{equation*}
Let $n$ be the number of data points selected from this trajectory, then an estimator of $I$ is constructed as
\begin{equation*}
\hat{I} = \frac{1}{n}\sum_{i=1}^nh(\bold{x}_{\frac{iT_M}{n}})
\end{equation*}
\subsection{Implementation}
\noindent The main difficult to implement BPS and GBPS is to simulate event time, which follows a Poisson process. The common techniques are based on the thinning and superposition theorems of Poisson process.
\begin{theorem}[Superposition Theorem \cite{kingman1993poisson}]
Let $\Pi_1, \Pi_2, \cdots,$ be a countable collection of independent Poisson processes on state space $\mathbb{R}^+$ and let $\Pi_n$ have rate $\lambda_n(t)$ for each $n$. If $\displaystyle{\sum_{n=1}^{\infty}}\lambda_n(t) < \infty$ for all $t$, then the superposition
\begin{equation*}
\Pi = \bigcup_{n=1}^{\infty}\Pi_n
\end{equation*}
is a Poisson process with rate
\begin{equation*}
\lambda(t) =  \sum_{n=1}^{\infty}\lambda_n(t)
\end{equation*}
\end{theorem} 
\begin{theorem}[Thinning Theorem \cite{lewis1979simulation}]
Let $\lambda: \mathbb{R}^+\rightarrow\mathbb{R}^+$ and $\Lambda: \mathbb{R}^+\rightarrow\mathbb{R}^+$ be continuous such that $\lambda(t) \leq \Lambda(t)$ for all $t\geq0$. Let $\tau^1, \tau^2, \cdots,$ be the increasing finite or infinite sequence of points of a Poisson process with rate $\Lambda(t)$. For all $i$, delete the point $\tau^i$ with probability $1 - \lambda(t)/\Lambda(t)$. Then the remaining points $\tilde{\tau}^1, \tilde{\tau}^2, \cdots$ form a non-homogeneous Poisson process with rate $\lambda(t)$.
\end{theorem}
\noindent In GBPS, from a given state $(\bold{x}, \bold{v})$, the associated Poisson process $\Pi_{\bold{x}, \bold{v}}$ has a rate function $\lambda(t) = \lambda(\bold{x}+t\bold{v}, \bold{v})$. With the help of the above two theorems, we can simulate a sample from $\Pi_{\bold{x}, \bold{v}}$ feasibly. \\
\\
Let $\eta(t) = \int_0^t\lambda(s)ds$, then the first event time, $\tau$, of Poisson process $\Pi$, whose rate function is $\lambda(t)$, satisfies
\begin{equation*}
\mathbb{P}(\tau > u) = \mathbb{P}(\Pi\cap[0, u] = \emptyset) = \exp(-\eta(u))
\end{equation*}
By the inverse theorem, $\tau$ can be simulated with the help of a uniform variate $V\sim\mathcal{U}(0,1)$ via:
\begin{equation*}
\tau = \eta^{-1}(-\log(V))
\end{equation*}
If we can compute $\eta^{-1}$ analytically, it is easy to simulate the event times. Otherwise, the simulations commonly depend on the superposition and thinning theorems. 
\subsection{GBPS with Sub-sampling in Big Data}
\noindent In Bayesian analysis, we suppose the observations $\{y_1, y_2, \cdots, y_N\}$ are i.i.d. samples from some distribution in the family $\{\mathbb{P}_{\bold{x}},\bold{x}\in\mathbb{R}^d\}$ and let $\mathbb{P}_{\bold{x}}$ admit the density $p_{\bold{x}}$ with respect to the Lebesgue measure on $\mathbb{R}^d$. Given a prior $\pi_0(\bold{x})$ over the parameter $\bold{x}$, the posterior is 
\begin{equation*}
\pi(\bold{x}) \myeq \pi(\bold{x}|y_1, \cdots, y_N) \propto \pi_0(\bold{x})\prod_{n=1}^Np_{\bold{x}}(y_n)
\end{equation*}
Traditional MCMC algorithms (with MH step) are difficult to scale for large data set, since each MH step needs to sweep over the whole data set. However, as indicated in \cite{bierkens2016zig}, PDMP may be super-efficient by using sub-sampling to simulate samples from the target distribution if we can give a tight upper bound of the rate function. In GBPS, we only use the gradient of the logarithm of the target distribution, which means we can simulate the posterior by knowing it up to a constant. Besides, we can give an unbiased estimator of the gradient of the logarithm of the posterior by using its sum structure to simulate the posterior exactly:
\begin{equation*}
\widehat{\nabla\log\pi(\bold{x})} =N\nabla\log\pi_I(\bold{x})= \nabla\log\pi_0(\bold{x}) + N\nabla_{\bold{x}}\log p_{\bold{x}}(y_I), \quad I\sim \mathcal{U}\{1,2,\cdots, N\}
\end{equation*}
\\\noindent In Algorithm \ref{alg:subsampling}, we show the workflow of the implementation of subsampling in GBPS. Notice that $\lambda(\Delta, \bold{v}_{i-1})$ equals to $\lambda(\bold{x}, \bold{v}_{i-1})$ in which $\nabla\log\pi(\bold{x})$ is replaced by $\Delta$. $\Lambda(t)$ is an upper bound of $\lambda(\bold{x}, \bold{v})$.
\begin{algorithm}[h]
   \caption{Subsampling version}
   \label{alg:subsampling}
\begin{algorithmic}
\STATE {\bfseries Initialize:} $\bold{x}_0, \bold{v}_0, T_0 = 0$.
\FOR{ $i = 1, 2, 3, \cdots$}
    \STATE Generate $\tau \sim PP(\Lambda(t))$
    \STATE $T_i \leftarrow T_{i-1} + \tau$
    \STATE $\bold{x}_i \leftarrow \bold{x}_{i-1} + \tau\bold{v}_{i-1}$
    \STATE $I\sim \mathcal{U}( \{1, \cdots, N\} )$
    \STATE $\Delta \leftarrow N\nabla\log{\pi_I(\bold{x}_i)} $
    \STATE $q \leftarrow \lambda(\Delta, \bold{v}_{i-1})/\Lambda(\tau)$
    \STATE $u \sim \mathcal{U}(0,1)$
    \IF {$u \leq q$}
        \STATE $\bold{v}_i \leftarrow Q(d\bold{v}|\Delta, \bold{v}_{i-1})$
    \ELSE
        \STATE $\bold{v}_i \leftarrow \bold{v}_{i-1}$
    \ENDIF
    \ENDFOR
   \end{algorithmic}
\end{algorithm}
\section{Numerical simulations}
\noindent In this section, we apply GBPS algorithm on three numerical experiments. Example 1 shows that reducibility problem appears in isotropic Gaussian distribution for BPS without refreshment but is not encountered by GBPS.  In Example 2, we can find that GBPS works well on multimode distributions and with similar performance with BPS. Finally, we present the GBPS with sub-sampling on Bayesian logistic model.\\
\\
\noindent \textbf{Example 1}: (isotropic Gaussian distribution) In this example, we show the reducibility problem of BPS without refreshment. The target distribution is 
\begin{equation*}
\pi(\bold{x}) = \frac{1}{2\pi}\exp\left\{-\frac{x_1^2+x_2^2}{2}\right\}
\end{equation*}
First we apply the BPS without reference Poisson process and show its reducibility in Figure \ref{ReducibleBPS}. 
\begin{figure}[http]
\center
\includegraphics[width=10cm]{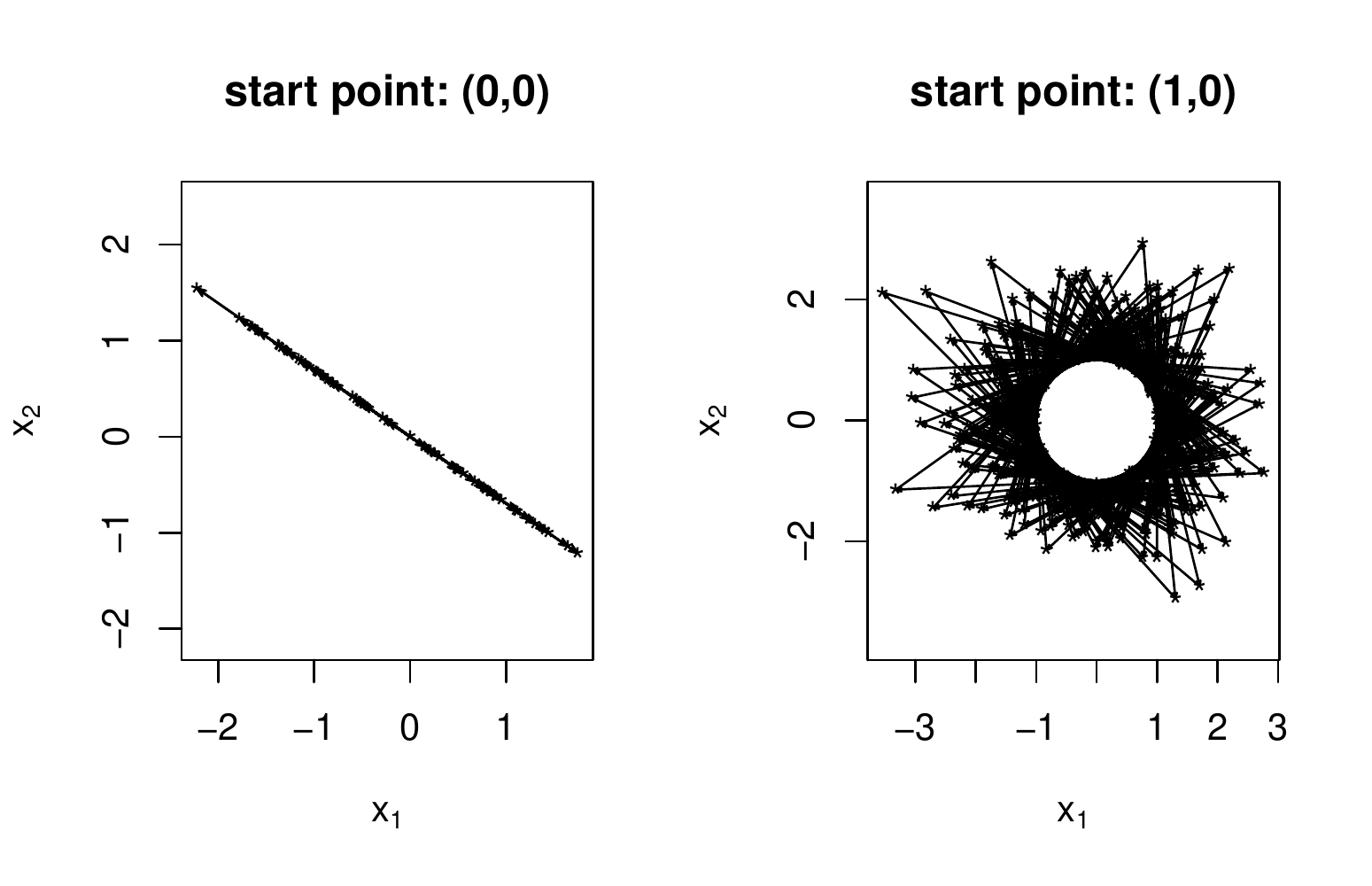}
\caption{Reducibility problem in isotropic Gaussian distributions: (left) the first 50 segments of a BPS path without refreshment which starts from the center of the Gaussian distribution, the trajectory is on a line; (right) the first 500 segments of another BPS path with $\lambda^{\text{ref}}=0$ starting from an point except the center, the trajectory cannot explore the center area.}
\label{ReducibleBPS}
\end{figure}
Compared with BPS without refreshment, GBPS is irreducible, shown in Figure \ref{BGPSIrreducible}.
\begin{figure}[http]
\center
\includegraphics[width=10cm]{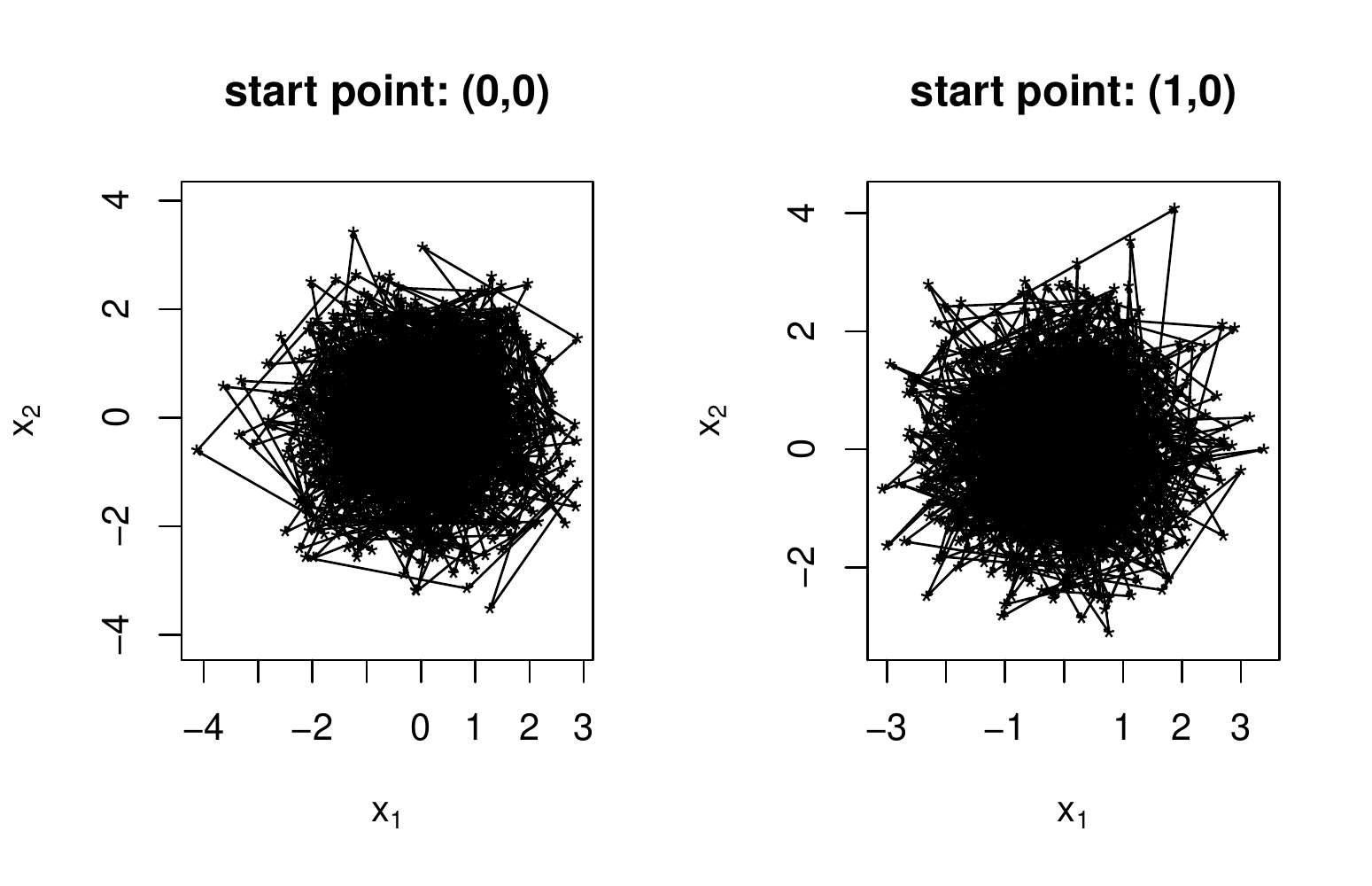}
\caption{GBPS is irreducible in isotropic Gaussian distribution. (left) the first 1000 segments of a GBPS path which starts from the center of the Gaussian distribution; (right) the first 1000 segments of another GBPS path starting from an point except the center}
\label{BGPSIrreducible}
\end{figure}
\\
\\Secondly, we compare the performance of GBPS and BPS with refreshment. For BPS, we set $\lambda^{\text{ref}} = \{0.01, 0.1, 0.2, 0.5, 1\}$. Each method is run 50 times and each sampled path has length $10^4$. For each path, we sample $10^4$ points with length gap $1$. Figure \ref{Error_Isotropic} shows the errors of the first and second moments of each component and Figure \ref{Ess_wasserstein} presents the errors in terms of Wasserstein-2 distance with respect to the target distribution and the effective sample size of each method. \\
\\
For BPS, we need to tune the rate of reference Poisson process to balance the efficiency and accuracy. Event though BPS is ergodic for every positive refreshment rate $\lambda^{\text{ref}}$ in theory, the value of $\lambda^{\text{ref}}$ matters in implementation. The smaller the refreshment rate, the larger the effective sample size (high efficiency), the more slowly the chain mixes. The larger the refreshment rate, the smaller the effective sample size (low efficiency), the faster the chain mixes. However, when the refreshment rate is extremely large or small, BPS will produce chains approximating the target distribution poorly. On the other hand, there is no hyper-parameter to tune in GBPS, which incorporates the randomness of BPS in refreshment into transition dynamics. Compared to BPS with different refreshment rates, GBPS performs modest in terms of the first and second moments of each component. In terms of Wasserstein-2 distance, GBPS outperforms BPS. 
\begin{figure}[http]
\center
\includegraphics[width=10cm]{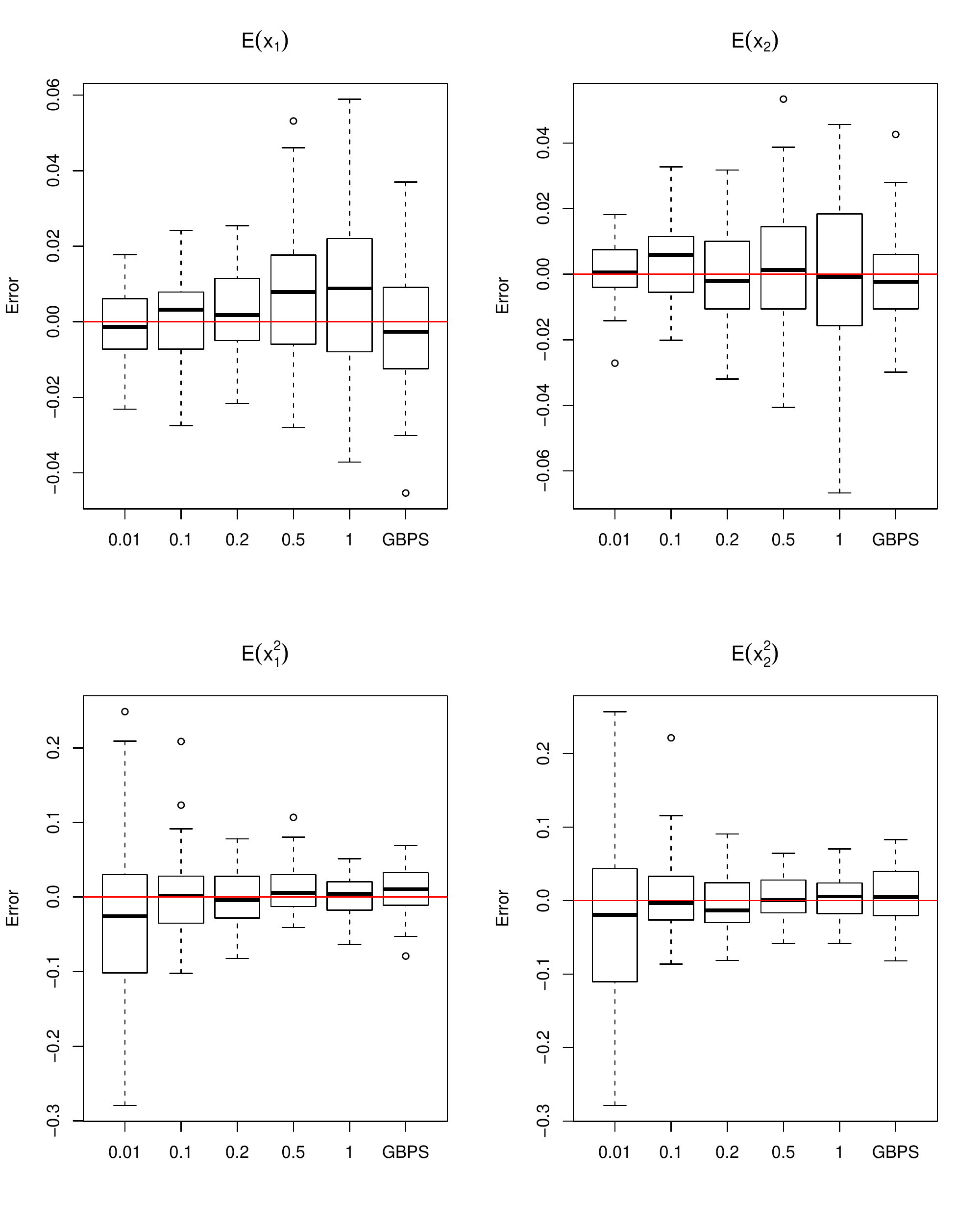}
\caption{Comparison between BPS and GBPS in isotropic Gaussian distribution. For each graph, the first five boxes represent the BPS method with different refreshment rates $\lambda^{\text{ref}} = \{0.01, 0.1, 0.2, 0.5, 1\}$. }
\label{Error_Isotropic}
\end{figure}
\begin{figure}[http]
\center
\includegraphics[width=10cm]{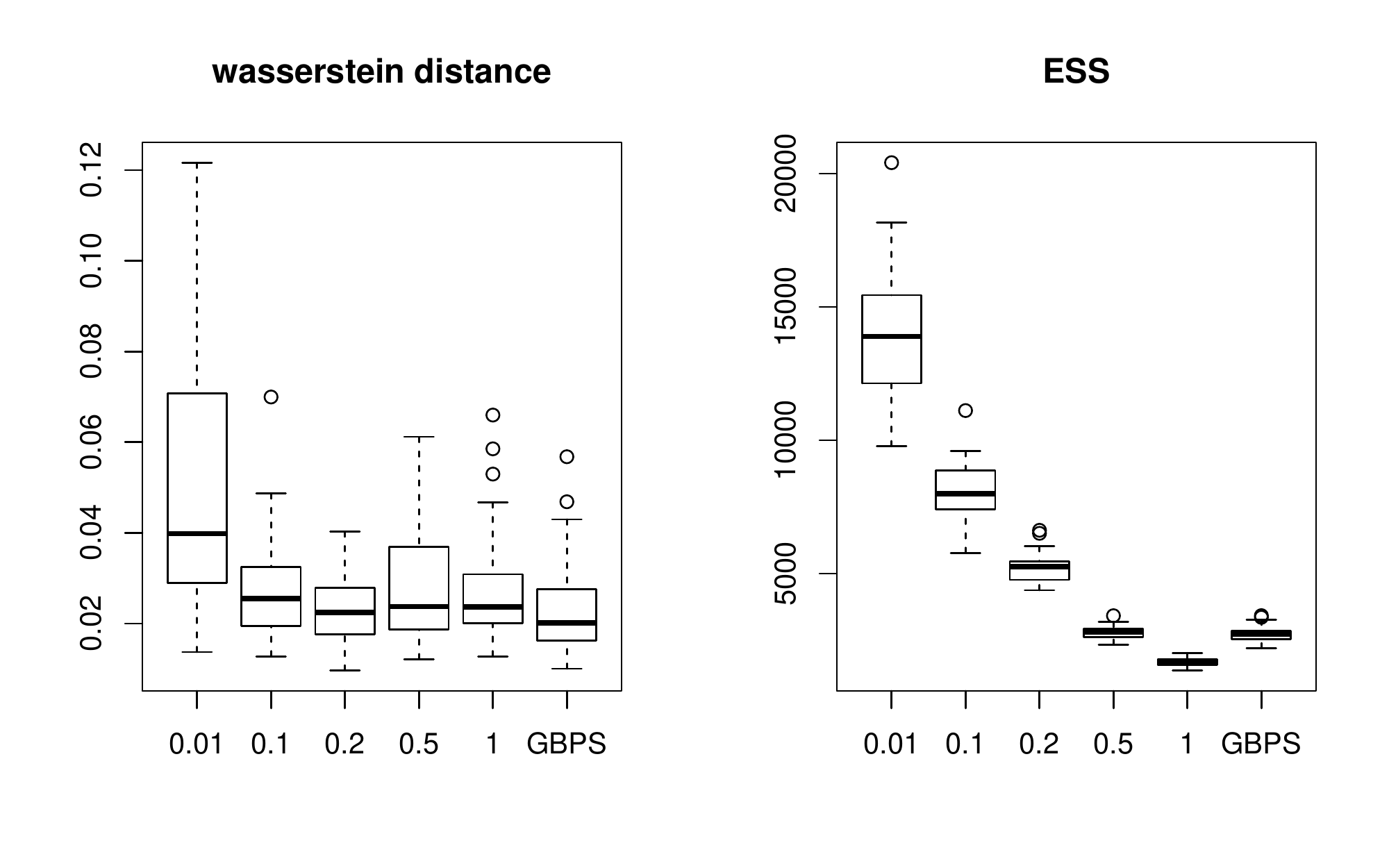}
\caption{Comparison between BPS and GBPS in isotropic Gaussian distribution in terms of Wasserstein distance and effective sample size. For each graph, the first five boxes represent the BPS method with different refreshment rates $\lambda^{\text{ref}} = \{0.01, 0.1, 0.2, 0.5, 1\}$. }
\label{Ess_wasserstein}
\end{figure}
\\
\\
\noindent\textbf{Example 2}: (mixture of Gaussian model) In this example, we show how to simulate the event time by using superposition and thinning theorems. The target is a mixture of Gaussian distributions:
\begin{equation*}
\pi(x_1,x_2) = \frac{p}{2\pi\sigma_1\sigma_2}\exp\left\{-\frac{(x_1-3)^2}{2\sigma_1^2} - \frac{x_2^2}{2\sigma_2^2}\right\} + \frac{1-p}{2\pi\sigma_3\sigma_4}\exp\left\{-\frac{x_1^2}{2\sigma_3^2}-\frac{(x_2-3)^2}{2\sigma_4^2}\right\}
\end{equation*}
In our experiment, we set $p = 0.5,  (\sigma_1, \sigma_2, \sigma_3, \sigma_4) = (1, 1.5, 2, 1)$. The gradient is 
\begin{equation*}
\begin{split}
\frac{\partial \pi(x_1,x_2)}{\partial x_1} &= \frac{p}{2\pi\sigma_1\sigma_2}\exp\left\{-\frac{(x_1-3)^2}{2\sigma_1^2}- \frac{x_2^2}{2\sigma_2^2}\right\}\left(-\frac{(x_1-3)}{\sigma_1^2}\right)\\
&+\frac{1-p}{2\pi\sigma_3\sigma_4}\exp\left\{-\frac{x_1^2}{2\sigma_3^2}-\frac{(x_2-3)^2}{2\sigma_4^2}\right\}\left(-\frac{x_1}{\sigma_3^2}\right)
\end{split}
\end{equation*}
\begin{equation*}
\begin{split}
\frac{\partial \pi(x_1,x_2)}{\partial x_2} &= \frac{p}{2\pi\sigma_1\sigma_2}\exp\left\{-\frac{(x_1-3)^2}{2\sigma_1^2}- \frac{x_2^2}{2\sigma_2^2}\right\}\left(-\frac{x_2}{\sigma_2^2}\right)\\
&+\frac{1-p}{2\pi\sigma_3\sigma_4}\exp\left\{-\frac{x_1^2}{2\sigma_3^2}-\frac{(x_2-3)^2}{2\sigma_4^2}\right\}\left(-\frac{(x_2-3)}{\sigma_4^2}\right)
\end{split}
\end{equation*}
We can give an upper bound for the norm of the gradient of the logarithm of the target density function:
\begin{equation*}
\Vert \nabla \log\pi(x_1, x_2)\Vert_2 \leq \frac{\vert x_1-3\vert}{\sigma_1^2}+\frac{\vert x_1\vert}{\sigma_3^2}+\frac{\vert x_2\vert}{\sigma_2^2} +\frac{\vert x_2-3\vert}{\sigma_4^2}
\end{equation*}
Then an upper bound for $\lambda(\bold{x}, \bold{v})$ is given as 
\begin{equation*}
\lambda(\bold{x}, \bold{v}) \leq \left(\frac{\vert x_1-3\vert}{\sigma_1^2}+\frac{\vert x_1\vert}{\sigma_3^2}+\frac{\vert x_2\vert}{\sigma_2^2} +\frac{\vert x_2-3\vert}{\sigma_4^2}\right)*\Vert \bold{v}\Vert_2
\end{equation*}
By superposition, we need only focus on Poisson process whose rate function has such form: $\lambda(x,v) = \frac{\vert x-\mu\vert}{\sigma^2}$. Let $\lambda_s(x, v) = \lambda(x+sv,v) = \frac{\vert x+sv-\mu\vert}{\sigma^2}$. Define
\begin{equation*}
\eta(t) = \int_0^t\lambda_s(x,v)ds
\end{equation*}
$i)$: If $x>\mu, v >0$, 
\begin{equation*}
\begin{split}
\eta(t) &= \int_0^t\frac{(x-\mu) + sv}{\sigma^2}ds = \frac{\frac{1}{2}vt^2 + (x-\mu)t}{\sigma^2} = \frac{v}{2\sigma^2}\left(t^2 + \frac{2(x-\mu)}{v}t\right) \\
&= \frac{v}{2\sigma^2}\left[\left(t + \frac{(x-\mu)}{v}\right)^2 - \frac{(x-\mu)^2}{v^2}\right]
\end{split}
\end{equation*}
\begin{equation*}
\eta^{-1}(z) = \sqrt{\frac{2\sigma^2z}{v} +\frac{(x-\mu)^2}{v^2} } - \frac{(x-\mu)}{v}
\end{equation*}
$ii):$ If $x<\mu, v < 0$, then
\begin{equation*}
\begin{split}
\eta(t) &= \int_0^t\frac{-(x-\mu) - sv}{\sigma^2}ds = \frac{-\frac{1}{2}vt^2 - (x-\mu)t}{\sigma^2} = -\frac{v}{2\sigma^2}\left(t^2 + \frac{2(x-\mu)}{v}t\right) \\
&= -\frac{v}{2\sigma^2}\left[\left(t + \frac{(x-\mu)}{v}\right)^2 - \frac{(x-\mu)^2}{v^2}\right]
\end{split}
\end{equation*}
\begin{equation*}
\eta^{-1}(z) = \sqrt{-\frac{2\sigma^2z}{v} +\frac{(x-\mu)^2}{v^2} } - \frac{(x-\mu)}{v}
\end{equation*}
$iii):$ If $x > \mu, v \leq0$: 
\begin{equation*}
\eta\left(-\frac{x-\mu}{v}\right) = \int_0^{-\frac{x-\mu}{v}}\frac{sv+(x-\mu)}{\sigma^2}ds = -\frac{(x-\mu)^2}{2v\sigma^2}
\end{equation*}
\begin{enumerate}
  \item If $z > -\frac{(x-\mu)^2}{2v\sigma^2}$:  $t_0 = -\frac{x-\mu}{v}$
           \begin{equation*}
           -\frac{(x-\mu)^2}{2v\sigma^2} +\int_0^t-\frac{sv}{\sigma^2}ds = z, \quad t = \sqrt{-\frac{2\sigma^2z}{v} - \frac{(x-\mu)^2}{v^2}}
           \end{equation*}
           \begin{equation*}
           \eta^{-1}(z) = \sqrt{-\frac{2\sigma^2z}{v} - \frac{(x-\mu)^2}{v^2}} + \left(-\frac{x-\mu}{v}\right)
           \end{equation*}
  \item If $z \leq -\frac{(x-\mu)^2}{2v\sigma^2}$:
          \begin{equation*}
          \int_0^t\frac{sv+(x-\mu)}{\sigma^2}ds = \frac{v}{2\sigma^2}\left(t^2 + \frac{2(x-\mu)}{v}t\right) = z
          \end{equation*}
          \begin{equation*}
          \eta^{-1}(z) = -\sqrt{\frac{2\sigma^2z}{v} + \frac{(x-\mu)^2}{v^2}} + \left(-\frac{x-\mu}{v}\right)
          \end{equation*}
\end{enumerate}
$iv):$ If $x \leq \mu, v >0$: 
\begin{equation*}
\eta\left(-\frac{x-\mu}{v}\right) = \int_0^{-\frac{x-\mu}{v}}\frac{-sv-(x-\mu)}{\sigma^2}ds = \frac{(x-\mu)^2}{2v\sigma^2}
\end{equation*}
 \begin{enumerate}
  \item If $z > \frac{(x-\mu)^2}{2v\sigma^2}$:  $t_0 = -\frac{x-\mu}{v}$
           \begin{equation*}
           \frac{(x-\mu)^2}{2v\sigma^2} +\int_0^t\frac{sv}{\sigma^2}ds = z, \quad t = \sqrt{\frac{2\sigma^2z}{v} - \frac{(x-\mu)^2}{v^2}}
           \end{equation*}
           \begin{equation*}
           \eta^{-1}(z) = \sqrt{\frac{2\sigma^2z}{v} - \frac{(x-\mu)^2}{v^2}} + \left(-\frac{x-\mu}{v}\right)
           \end{equation*}
  \item If $z \leq \frac{(x-\mu)^2}{2v\sigma^2}$:
          \begin{equation*}
          \int_0^t\frac{-sv-(x-\mu)}{\sigma^2}ds = -\frac{v}{2\sigma^2}\left(t^2 + \frac{2(x-\mu)}{v}t\right) = z
          \end{equation*}
          \begin{equation*}
          \eta^{-1}(z) = -\sqrt{-\frac{2\sigma^2z}{v} + \frac{(x-\mu)^2}{v^2}} + \left(-\frac{x-\mu}{v}\right)
          \end{equation*}
\end{enumerate}
In Figure \ref{multimode_samples}, we show the trajectory of the simulated GBPS path and associated samples. Figure \ref{multimode_density} shows the marginal density functions of the target distribution. In Figure \ref{BPSvsGBPS_Multi}, we compare the performance of BPS and GBPS. For BPS, we set $\lambda^{\text{ref}} = 0.01, 0.1, 1$. We sample 50 paths with length $10000$ for each method and take $10000$ points from each path with gap 1 to form samples. Empirically, BPS is ergodic over this example. With the increase of $\lambda^{\text{ref}}$, the refreshment occurs more frequently, which reduces the performance of BPS. Even though GBPS has worse performance, compared to BPS with some refreshment rates, it is quite reliable and has no parameter to tune. 
\begin{figure}[http]
\center
\includegraphics[width=10cm]{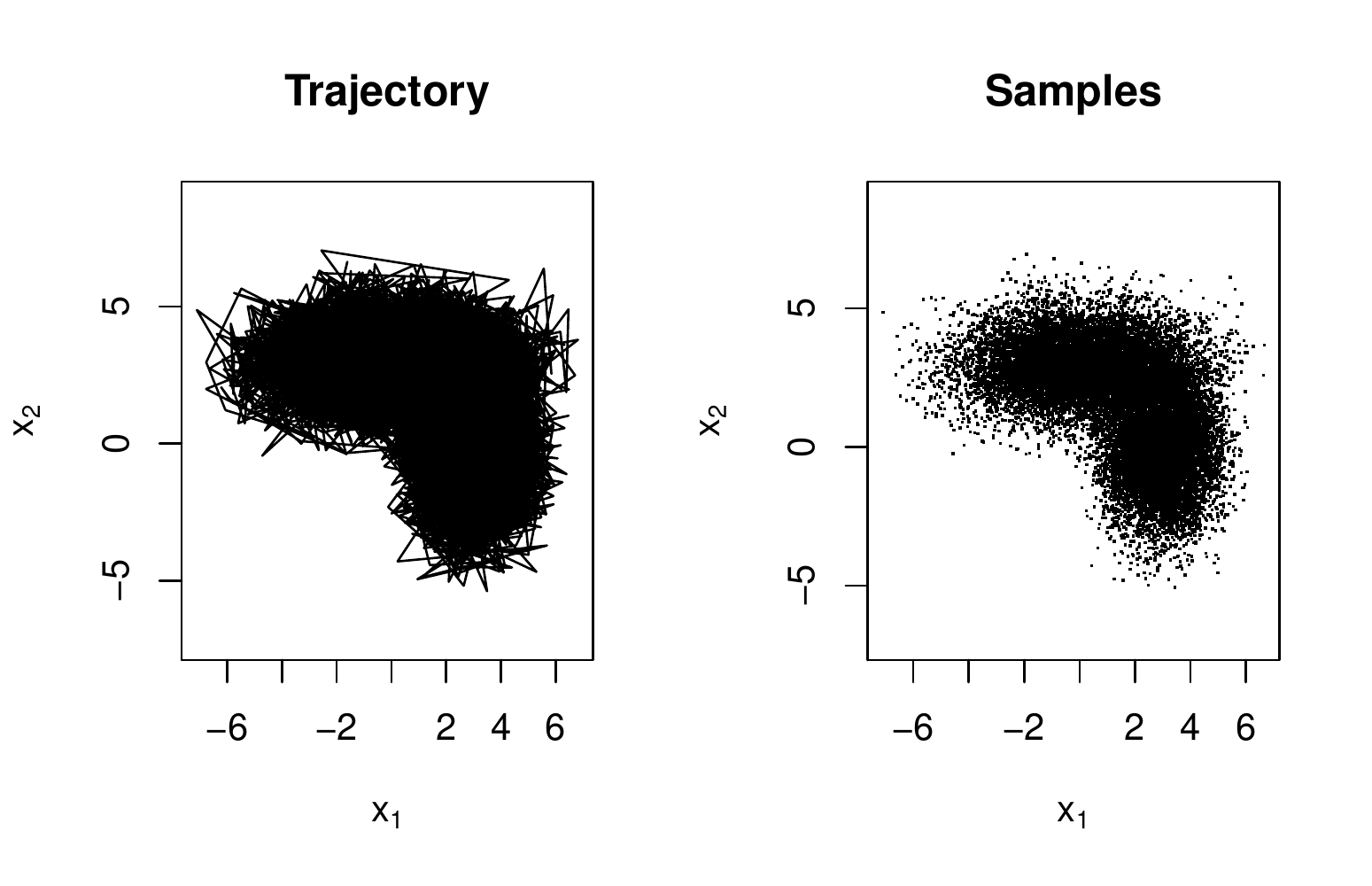}
\caption{The trajectory and samples from a GBPS path.}
\label{multimode_samples}
\end{figure}
\begin{figure}[!http]
\center
\includegraphics[width=10cm]{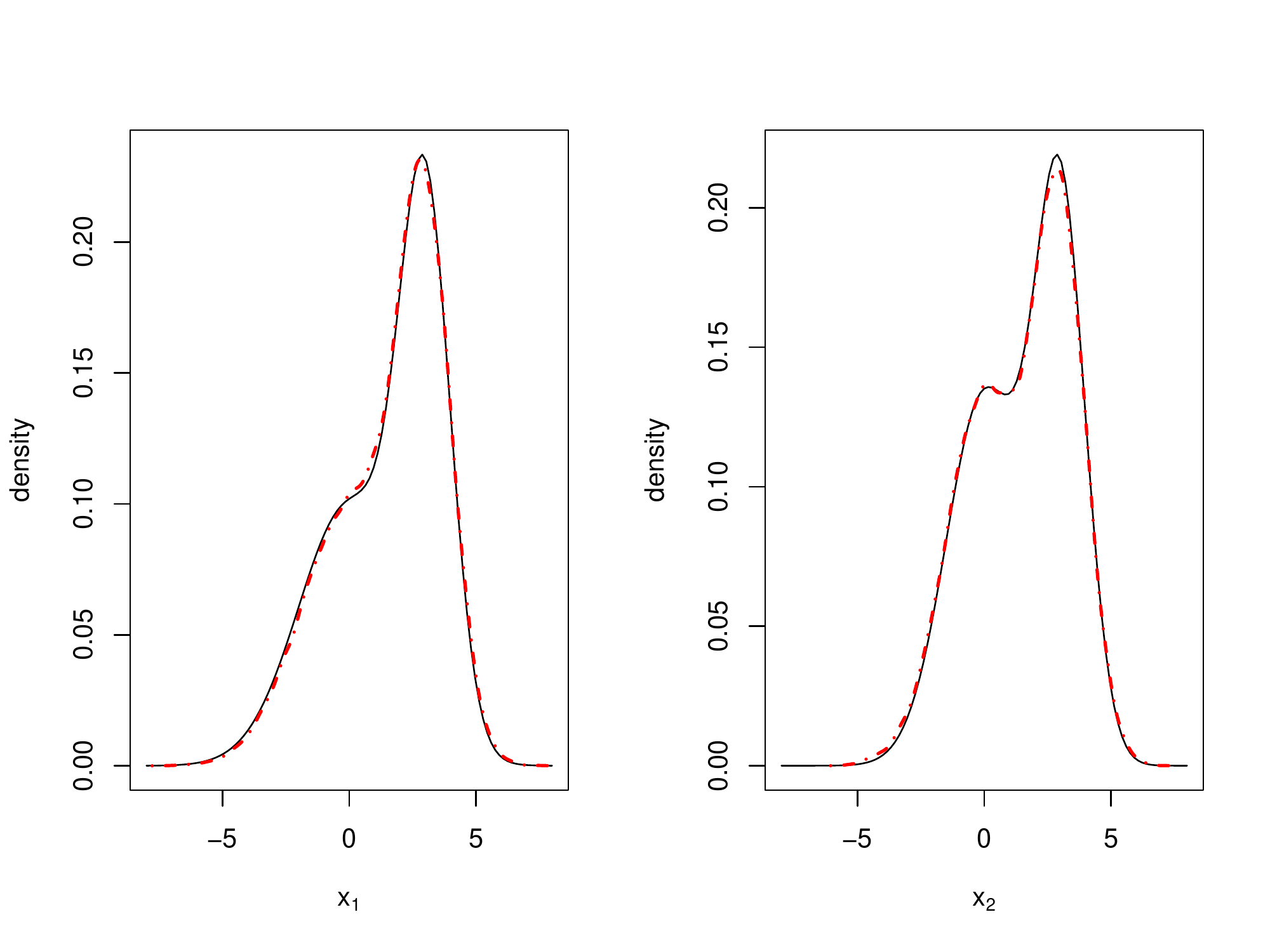}
\caption{The marginal density functions: the black solid lines are true marginal density, the red dotted lines are from a GBPS path.}
\label{multimode_density}
\end{figure}
\begin{figure}[!http]
\center
\includegraphics[width=10cm]{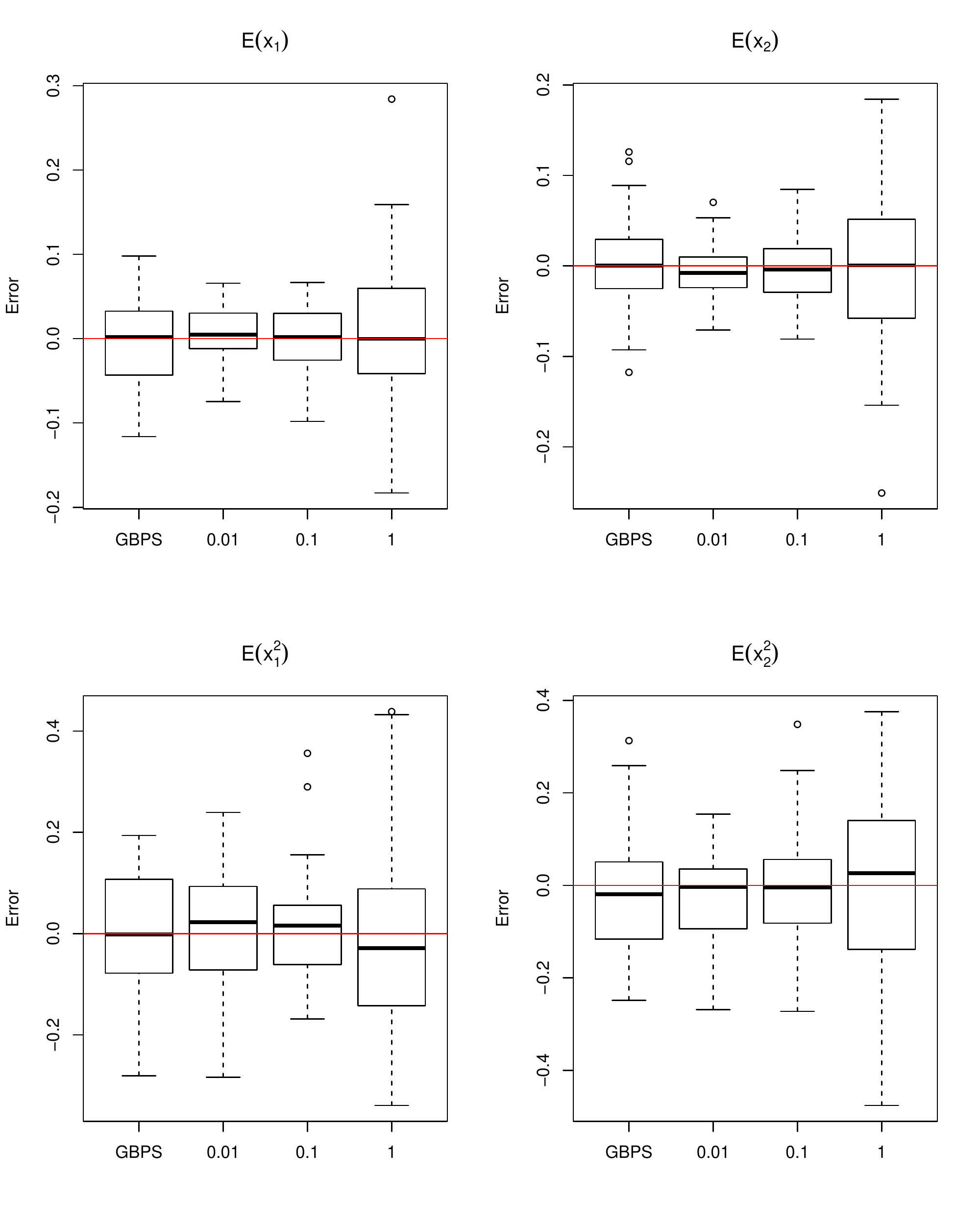}
\caption{Comparison between GBPS and BPS: for each graph, the last three boxes represent BPS with $\lambda^{\text{ref}} = 0.01, 0.1, 1$.}
\label{BPSvsGBPS_Multi}
\end{figure}
\\
\\
\textbf{Example 3}: (Bayesian Logistic Model) For the Bayesian logistic model, we suppose $\bold{x}\in\mathbb{R}^d$ be the parameters and $(y_i, z_i)$, for $i  = 1, 2, \cdots, N$ be the observations,  where $y_i\in\mathbb{R}^d, z_i\in\{0, 1\}$, then
\begin{equation*}
\mathbb{P}(z_i =1 | y_i, \bold{x}) = \frac{1}{1+\exp\{-\sum_{\ell=1}^dy_i^{\ell}x_{\ell}\}}
\end{equation*}
Choosing the improper prior, then the posterior is 
\begin{equation*}
\pi(\bold{x})\propto\prod_{j=1}^N\frac{\exp\{z_j\sum_{\ell=1}^dy_j^{\ell}x_{\ell}\}}{1+\exp\{\sum_{\ell=1}^dy_j^{\ell}x_{\ell}\}}
\end{equation*}
for $k = 1, \cdots, d$, the partial derivative is
\begin{equation*}
\frac{\partial}{\partial x_{k}}\log\pi(\bold{x}) = \sum_{j=1}^N\left[z_j -\frac{\exp\{z_j\sum_{\ell=1}^dy_j^{\ell}x_{\ell}\}}{1+\exp\{\sum_{\ell=1}^dy_j^{\ell}x_{\ell}\}} \right]y_{j}^k
\end{equation*}
Then, they are bounded by
\begin{equation*}
\left| \frac{\partial\log\pi(\bold{x})}{\partial x_k}\right| \leq\sum_{j=1}^N \left| y^k_j\right|
\end{equation*}
and the bounded rate for Poisson process is 
\begin{equation*}
\lambda^+ = \max_{1\leq k \leq d}\sum_{j=1}^N \left| y_j^k\right|
\end{equation*}
In our experiment, we set $d = 5$, $N = 100$ and use $10$ observations for subsampling at each iteration. Figure \ref{LogitDensity} shows the marginal density functions for each component of parameters.
\begin{figure}[!http]
\center
\includegraphics[width=10cm]{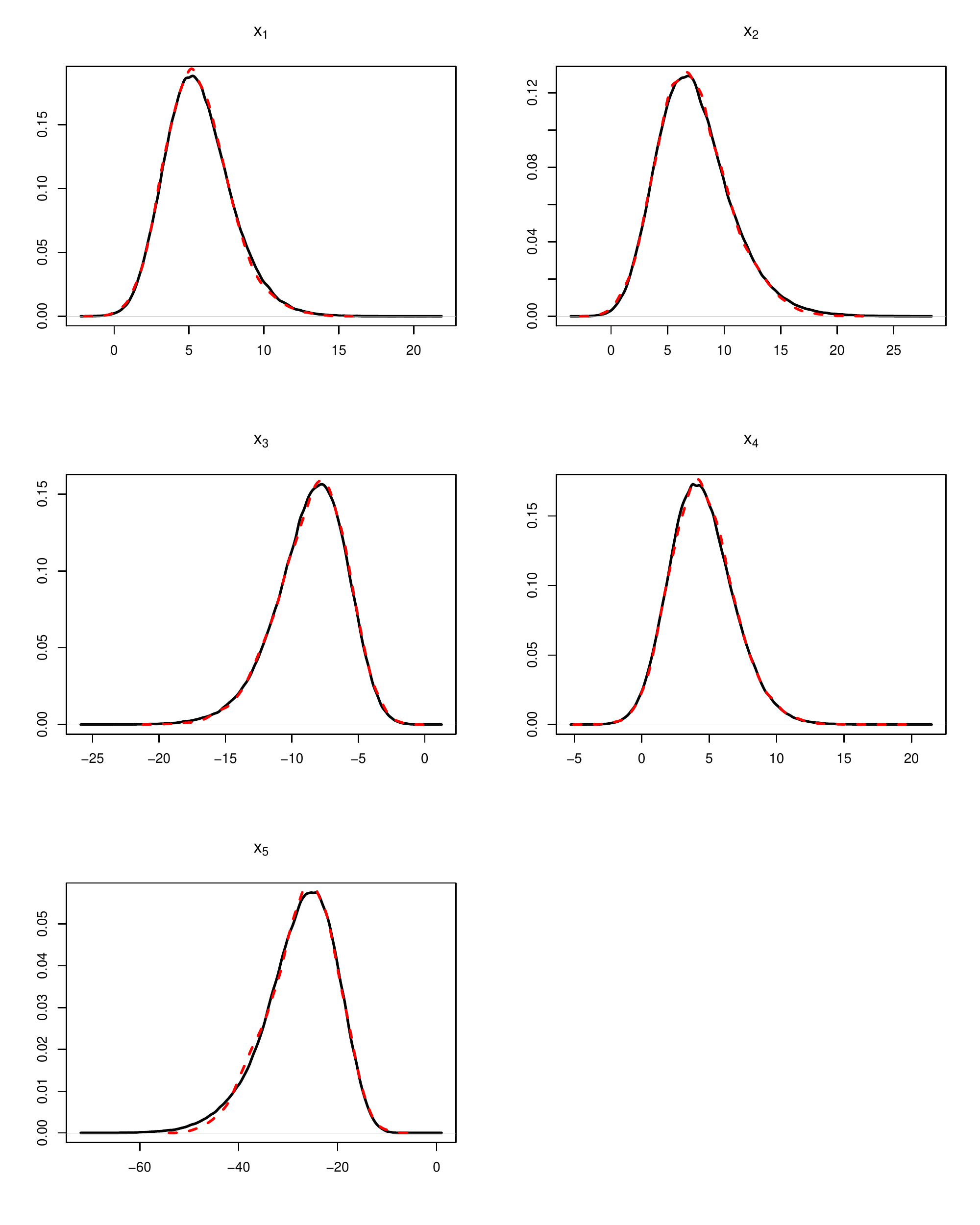}
\caption{The marginal density functions: the black solid lines are marginal density of MH algorithm, which are used as benchmark. The red dotted lines are from a GBPS path.}
\label{LogitDensity}
\end{figure}
\section{Conclusion}
In this article, we generalize the bouncy particle sampler in terms of its transition dynamics. Our method --- Generalized Bouncy Particle Sampler (GBPS) --- can be regarded as a bridge between bouncy particle sampler and zig-zag process sampler. Compared with bouncy particle sampler, GBPS changes the direction velocity according to some distribution at event time. However, compared with zig-zag process sampler, GBPS can be regarded as attaching a moving coordinate system on the state space of $(\bold{x}, \bold{v})$, instead of using a fixed one as zig-zag process sampler. One main advantage of GBPS, compared to BPS, is that it has no parameter to tune.\\
\\
Throughout the whole paper, we suppose that the parameter space has no restrictions. In practice, it is often the case one encounters restricted parameter space problems. In such cases, we may transfer the restricted region into the whole Euclidean space by reparameterization techniques. Besides, \cite{bierkens2017piecewise} shows some methods to simulate over restricted space. Another problem of implementation of these methods is how to simulate event time from Poisson process efficiently. Generally, the simulations are based on superposition and thinning theorems. The upper bound of rate function is crucial. The tighter the upper bound, the more efficient the simulation is. In Bayesian analysis for large data sets, if the upper bound is $\mathcal{O}(N^{\alpha})$, then the effective sample size per likelihood computation is $\mathcal{O}(N^{-(1/2+\alpha)})$. If $\alpha < 1/2$, then both BPS and GBPS will be more efficient than traditional MCMC methods.\\
\\
Exploring several simulation settings, we find that reducibility problem just appears in isotropic Gaussian distribution or in distributions who admit isotropic Gaussian distribution as their component for BPS. However, it is still an open question and needs to prove.
\newpage
\section{Appendix}
\noindent In this appendix, we prove the ergodicity of GBPS. For simplicity, we introduces the following notations. $\pi$ denotes our target distribution, $\text{Vol}$ denotes the Lebesgue measure over Euclidean space. 
\begin{equation*}
\bold{S}(\bold{v}) = \frac{\bold{v}}{\Vert\bold{v}\Vert_2}.
\end{equation*}
\begin{equation*}
S^{\perp}(\bold{x}, \bold{v}) = \{\bold{x}': \langle \bold{v}, \nabla\log\pi(\bold{x})\rangle \times \langle \bold{S}(\bold{x}'-\bold{x}), \nabla\log\pi(\bold{x})\rangle < 0\},
\end{equation*}
\noindent \textbf{Assumption 1}: For any two points $\bold{x}_1, \bold{x}_2 \in \mathbb{R}^d$ and any velocity $\bold{v}\in\mathbb{R}^d, \Vert \bold{v} \Vert_2 = 1$, there exists $t > 0$, such that 
\begin{equation*}
\bold{x}_2\in S^{\perp}(\bold{x}_1 + t\bold{v}, \bold{v})
\end{equation*}
\\
\textbf{Assumption 2}: $\text{Vol}\left(\left\{\bold{x}: \nabla\log\pi(\bold{x}) = \bold{0}\right\} \right)= 0$.\\
\\
\textbf{Remark:} Actually, Assumption 2 can be removed without influence on the correctness of Theorem 2  via a similar proof with that in the following one.\\
\\
\textbf{Lemma 1:} The Markov chain, $\bold{z}'_t = (\bold{x}_t, \frac{\bold{v}_t}{\Vert \bold{v}_t\Vert_2})$ induced by GBPS, admits $\pi(\bold{x})\times \mathcal{U}(S_{d-1})$ as its invariant distribution. \\
\\
\textbf{Lemma 2:} For any $\bold{x}_0\in\mathbb{R}^d, \bold{v}_0\in\mathcal{U}(S_{d-1})$, and any open set $W\subset\mathcal{R}^d\times S_{d-1}$, there exists some positive $t>0$ such that
\begin{equation*}
P_{t}\left((\bold{x}_0, \bold{v}_0), W\right) > 0
\end{equation*}

\begin{proof}[Proof of Lemme 2]
For any open set $W\subset \mathcal{R}^d\times S_{d-1}$, there exist $\bold{x}^*\in\mathcal{R}^d$, $r_1>0$ and $V\subset S_{d-1}$, such that $B(\bold{x}^*, r_1)\times V \subset \mathcal{R}^d\times S_{d-1}$ and $\nabla\log\pi(\bold{x}^*) \neq \bold{0}$. 
For $(\bold{x}_0, \bold{v}_0)$, according to Assumption 1, there exist $t_1\in(0, \infty)$ and a positive constant $\delta_1$,  such that
\begin{equation*}
\langle\bold{v}_0, \nabla\log\pi(\bold{x}_t)\rangle < 0, \text{ for all } t\in[t_1-\delta_1, t_1+\delta_1],
\end{equation*}
\begin{equation*}
\bold{x}^*\in S^{\perp}(\bold{x}_t, \bold{v}_0)
\end{equation*}
By a minor transition of $t_1$, we can suppose that 
\begin{equation*}
\langle\bold{x}^* - (\bold{x}_0 + t_1\bold{v}_0), \nabla\log\pi(\bold{x^*})\rangle \neq 0
\end{equation*}
By selecting $\delta_1$ and $r_1$ small enough, we can get: \\
\\
(i) $\forall \bold{x}', \bold{x}''\in B(\bold{x}^*, r_1), \forall t\in[t_1-\delta_1, t_1+\delta_1]$, 
\begin{equation*}
\langle \bold{x}' - (\bold{x}_0 + t\bold{v}_0), \nabla\log\pi(\bold{x}'')\rangle
\end{equation*}
is always positive or negative;\\
\\
(ii) for each $t\in [t_1-\delta_1, t_1+\delta_1]$,
\begin{equation*}
B(\bold{x}^*, r_1) \subset S^{\perp}(\bold{x}_t, \bold{v}_0)
\end{equation*}
(iii) 
\begin{equation*}
\text{Vol}\left(\displaystyle{\bigcap_{t\in[t_1-\delta_1, t_1+\delta_1]}} \bold{S}\left(B\left(\bold{x}^*, r_1\right) - \bold{x}_t\right)\right) > 0
\end{equation*}
(iv)
\begin{equation*}
\bigcap_{t\in[t_1-\delta_1, t_1+\delta_1]}\bold{S}\left(B\left(\bold{x}^*,r_1\right) - \bold{x}_t\right)\subset \bigcap_{t\in[t_1-\delta_1, t_1+\delta_1]}S^{\perp}(\bold{x}_t, \bold{v}_0)
\end{equation*}
As a result, by (iv), there exists a positive constant, $p_1 = p_1(\bold{x}_0, \bold{v}_0, t_1, \delta_1, \bold{x}^*, r_1)$, such that, $\forall t\in[t_1-\delta_1, t_1+\delta_1]$, 
\begin{equation*}
Q\left(\bigcap_{t'\in[t_1-\delta_1, t_1+\delta_1]}\bold{S}\left(B(\bold{x}^*,r_1) - \bold{x}_{t'}\right)|\bold{x}_t, \bold{v}_0\right) > p_1
\end{equation*}
\textbf{Case 1}: If 
\begin{equation*}
\text{Vol}\left(\displaystyle{\bigcap_{t\in[t_1-\delta_1, t_1+\delta_1]}} \bold{S}\left(B\left(\bold{x}^*, r_1\right) - \bold{x}_t\right) \bigcap V\right) > 0,
\end{equation*}
then there is $t>0$, such that
\begin{equation*}
P_{t}\left((\bold{x}_0, \bold{v}_0), W\right) > 0
\end{equation*}
We illustrate Case 1 in Figure \ref{Case1}.
\begin{figure}[http]
\center
\includegraphics[width=10cm]{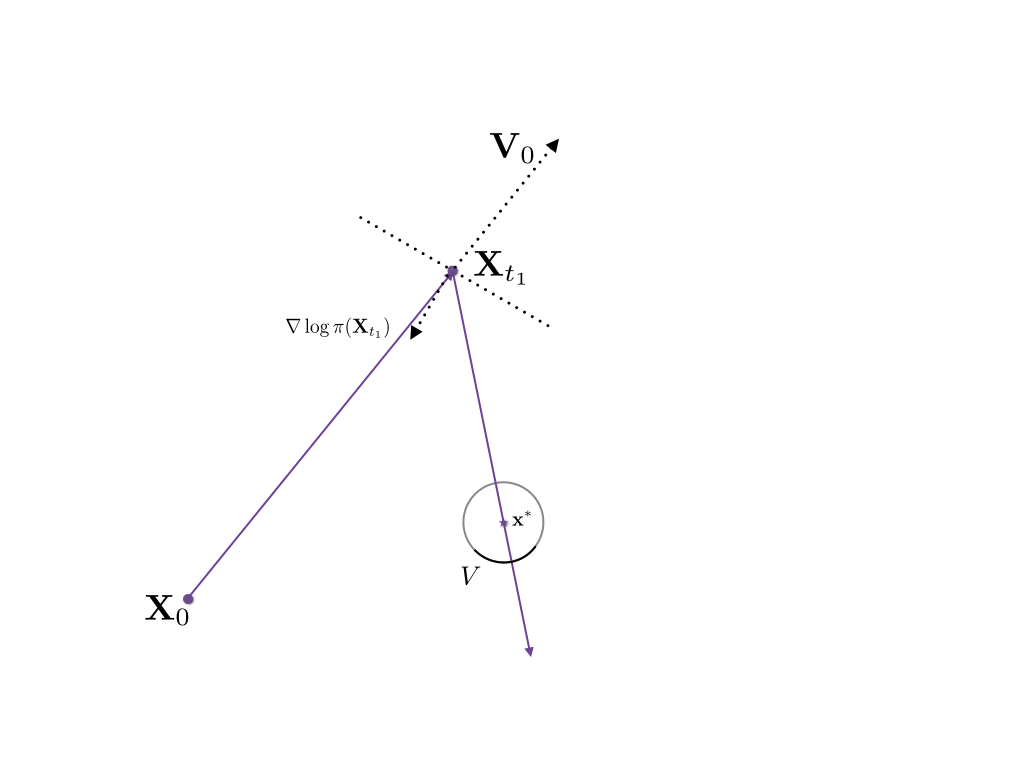}
\caption{Case1}
\label{Case1}
\end{figure}
\\
\textbf{Case 2}: If 
\begin{equation*}
\text{Vol}\left(\displaystyle{\bigcap_{t\in[t_1-\delta_1, t_1+\delta_1]}} \bold{S}\left(B\left(\bold{x}^*, r_1\right) - \bold{x}_t\right) \bigcap V\right) = 0
\end{equation*}
and 
\begin{equation*}
\langle \bold{S}(\bold{x}^* - \bold{x}_{t_1}), \nabla\log\pi(\bold{x}^*) \rangle < 0
\end{equation*}
By (i), we have $\forall \bold{x}', \bold{x}''\in B(\bold{x}^*, r_1), \forall t\in[t_1-\delta_1, t_1+\delta_1]$,
\begin{equation*}
\langle \bold{x}' - (\bold{x}_0 + t\bold{v}_0), \nabla\log\pi(\bold{x}'')\rangle < 0.
\end{equation*}
Denote
\begin{equation*}
\bold{M}_1 \quad=\bigcap_{\substack{\bold{x}\in B(\bold{x}^*,r_1), \bold{v}\in\bold{S}(\bold{x} - \bold{x}_t),\\ t\in[t_1-\delta_1, t_1+\delta_1]}}S^{\perp}(\bold{x}, \bold{v})
\end{equation*}
By decreasing $\delta_1, r_1$, we can have 
\begin{equation*}
\text{Vol}\left(\bold{M}_1\bigcap \bold{V}\right) > 0
\end{equation*}
or
\begin{equation*}
\text{Vol}\left(\bold{M}_1\bigcap \bold{V}\right) = 0,\quad \text{Vol}\left(\bold{M}_1\bigcap \bold{V}'\right) > 0, \text{ where } \bold{V}' = \bold{V}
\end{equation*}
\\
\textbf{Case 2.1}: If
\begin{equation*}
\text{Vol}\left(\bold{M}_1\bigcap \bold{V}\right) > 0
\end{equation*}
By selecting $r_2 < r_1$, such that the length of segment of the line across $\bold{x}_t$ and $\bold{x}'\in B(\bold{x}^*, r_2)$ in the ball $B(\bold{x}^*, r_1)$
is larger than some positive constant $c_1$. Then event occurs on each line across $\bold{x}_t$ and $\bold{x}'$ during the segment in the ball $B(\bold{x}^*, r_1)$ and the changed velocity traverses $\bold{V}$ with positive probability which is larger than some positive constant $p_2$. As a result, there exists some $t>0$ such that 
\begin{equation*}
P_{t}\left((\bold{x}_0, \bold{v}_0), W\right) > 0
\end{equation*}
See Figure \ref{Case2.1} for illustration.\\
\begin{figure}[http]
\center
\includegraphics[width=10cm]{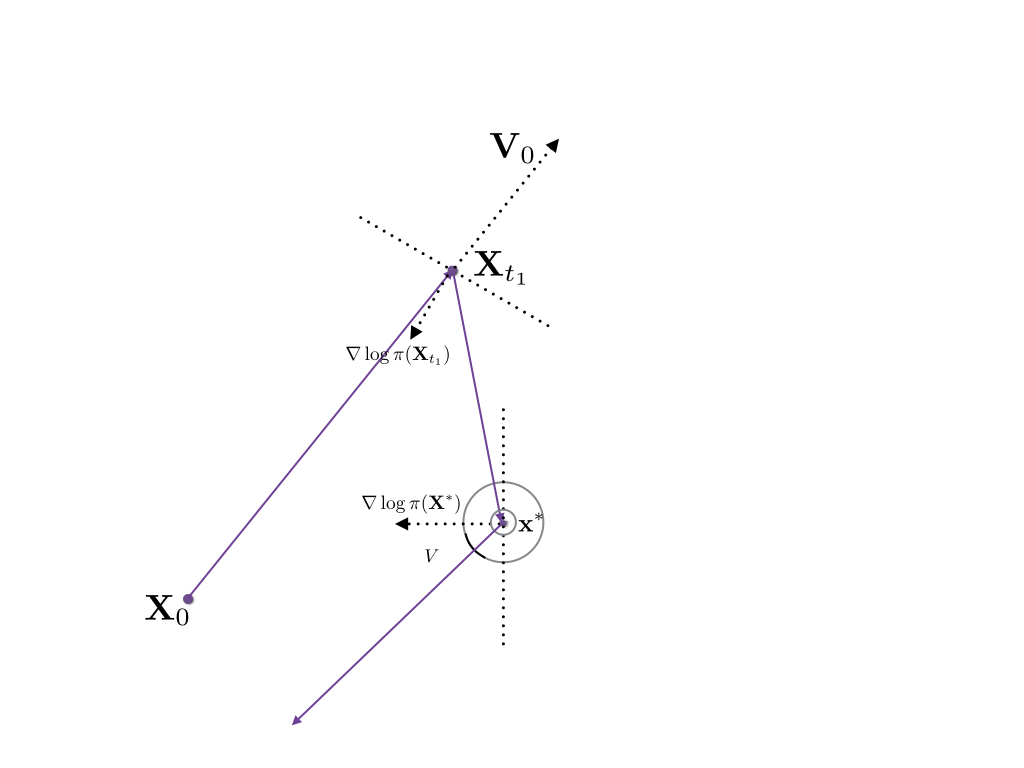}
\caption{Case 2.1}
\label{Case2.1}
\end{figure}
\\
\textbf{Case 2.2}: If 
\begin{equation*}
\text{Vol}\left(\bold{M}_1\bigcap \bold{V}\right) = 0,\quad \text{Vol}\left(\bold{M}_1\bigcap \bold{V}'\right) > 0, \text{ where } \bold{V}' = -\bold{V}
\end{equation*}
With the same treatment as that in Case 2.1, there exists $\bold{x}^{**}$ such that 
\begin{equation*}
\bold{S}(\bold{x}^{**} - \bold{x}^*)\in\bold{V}', \quad \text{and } \langle {S}(\bold{x}^{**} - \bold{x}^*), \nabla\log\pi(\bold{x}^{**})\rangle < 0
\end{equation*}
As a result, there exist two positive constants $r_3 > r_4$, such that, if necessary, decreasing $r_1$ to be small enough, 
\begin{equation*}
\bold{S}(\bold{x}'' - \bold{x}') \in \bold{V}', \text{ where } \bold{x}'\in B(\bold{x}^*, r_1),\quad\bold{x}''\in B(\bold{x}^{**}, r_3)
\end{equation*}
and 
\begin{equation*}
\text{Vol}\left(\bold{M}_1\bigcap\bold{M}_2\bigcap\bold{V}'\right)>0,
\end{equation*}
where 
\begin{equation*}
\bold{M}_2 = \bigcup_{\substack{\bold{x}'\in B(\bold{x}^*, r_1),\\\bold{x}''\in B(\bold{x}^{**}, r_4)}}\bold{S}(\bold{x}'' - \bold{x}').
\end{equation*}
Then, the event, that a particle begins from any point $\bold{x}'\in B(\bold{x}^*, r_1)$ with any velocity $\bold{v}'\in\bold{S}(B(\bold{x}^{**}, r_4) - \bold{x}')$, changes velocity in the region $B(\bold{x}^{**}, r_4)$ to $-\bold{M}_2$ and passes the area $B(\bold{x}^*, r_1)$, occurs with positive probability. As a result, there exists $t>0$ such that
\begin{equation*}
P_{t}\left((\bold{x}_0, \bold{v}_0), W\right) > 0
\end{equation*}
See Figure \ref{Case2.2} for illustration.
\begin{figure}[http]
\center
\includegraphics[width=10cm]{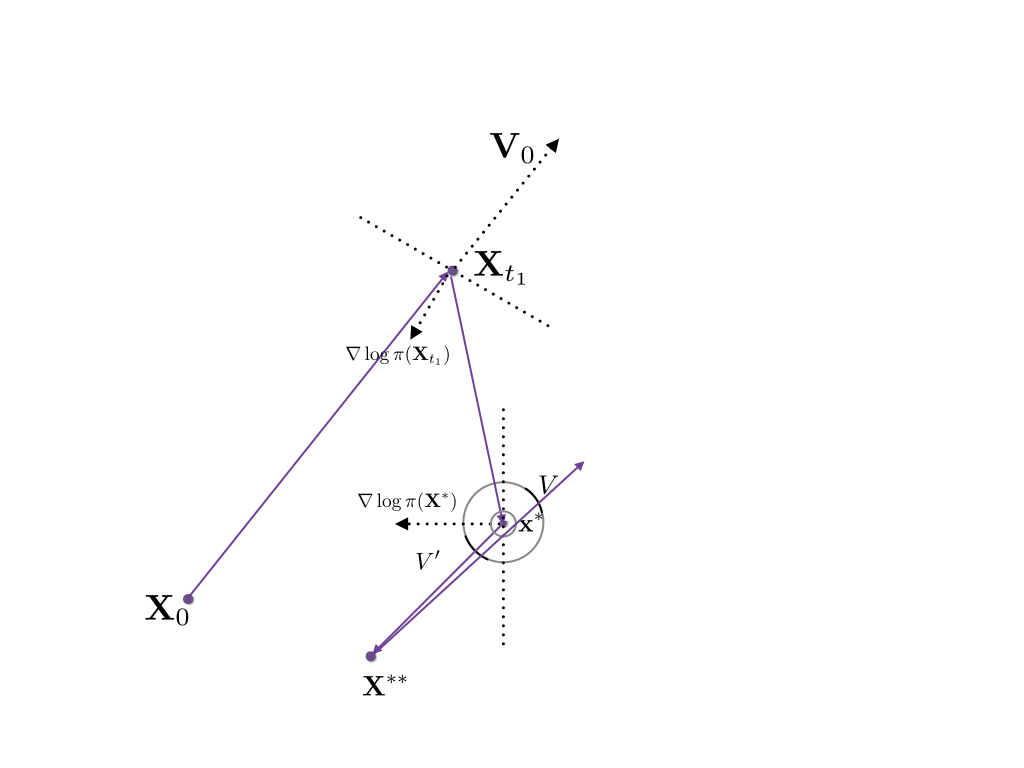}
\caption{Case 2.2}
\label{Case2.2}
\end{figure}
\\
\\\noindent \textbf{Case 3}: If
\begin{equation*}
\text{Vol}\left(\displaystyle{\bigcap_{t\in[t_1-\delta_1, t_1+\delta_1]}} \bold{S}\left(B\left(\bold{x}^*, r_1\right) - \bold{x}_t\right) \bigcap V\right) = 0
\end{equation*}
and 
\begin{equation*}
\langle \bold{S}(\bold{x}^* - \bold{x}_{t_1}), \nabla\log\pi(\bold{x}^*) \rangle > 0
\end{equation*}
By Assumption 1, there exists $\bold{x}^{**}$ on the line which passes $\bold{x}_{t_1}$ and $\bold{x}^*$, such that
\begin{equation*}
\langle \bold{S}(\bold{x}^{*} - \bold{x}_{t_1}), \nabla\log\pi(\bold{x}^{**})\rangle < 0. 
\end{equation*}
Then there exists $r_5$, such that, if necessary, decreasing $\delta_1, r_1$ to be small enough,
\begin{equation*}
\text{Vol}\left(\bold{M}_3\bigcap \bold{V}\right) > 0
\end{equation*}
or
\begin{equation*}
\text{Vol}\left(\bold{M}_3\bigcap \bold{V}\right) = 0\text{,  and  }\text{Vol}\left(\bold{M}_3\bigcap \bold{V}'\right) > 0, \text{ where } \bold{V}' = -\bold{V}
\end{equation*}
where
\begin{equation*}
\bold{M}_3 = \bigcap_{\substack{\bold{x}\in B(\bold{x}^*, r_1),\\ \bold{v}\in \bold{S}(\bold{x} - B(\bold{x}^{**}, r_5)) }}S^{\perp}(\bold{x}, \bold{v})
\end{equation*}
\\
\textbf{Case 3.1}: If 
\begin{equation*}
\text{Vol}\left(\bold{M}_3\bigcap \bold{V}\right) > 0
\end{equation*}
Then, the event that a particle begins from any point $\bold{x}_t, t\in[t_1-\delta_1, t1+\delta_1]$ with velocity $\bold{v}$, passes the region $B(\bold{x}^*, r_1)$, changes velocity in the region $B(\bold{x}^{**}, r_5)$ to $\bold{v}'\in\bold{S}(B(\bold{x}^*, r_1)-B(\bold{x}^{**}, r_5))$, reaches the area $B(\bold{x}^*, r_1)$ and changes velocity to $\bold{V}$ occurs with positive probability. As a result, we obtain the desired result in this case. See Figure \ref{Case3.1} for illustration.
\begin{figure}[http]
\center
\includegraphics[width=10cm]{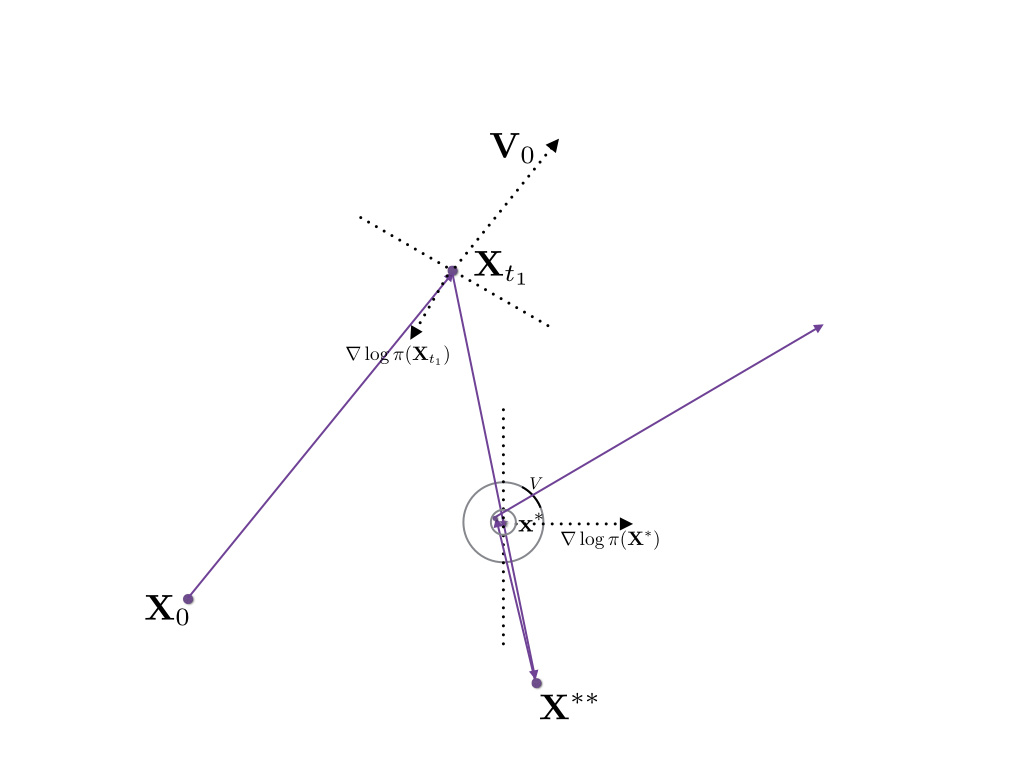}
\caption{Case 3.1}
\label{Case3.1}
\end{figure}
\\
\\\textbf{Case 3.2}: If
\begin{equation*}
\text{Vol}\left(\bold{M}_3\bigcap \bold{V}\right) = 0\text{,  and  }\text{Vol}\left(\bold{M}_3\bigcap \bold{V}'\right) > 0, \text{ where } \bold{V}' = -\bold{V}
\end{equation*}
By the similar treatment of Case 3.1 and Case 2.2, we can get the admired result. We give an illustration of Case 3.2 in Figure \ref{Case3.2}.
\begin{figure}[http]
\center
\includegraphics[width=10cm]{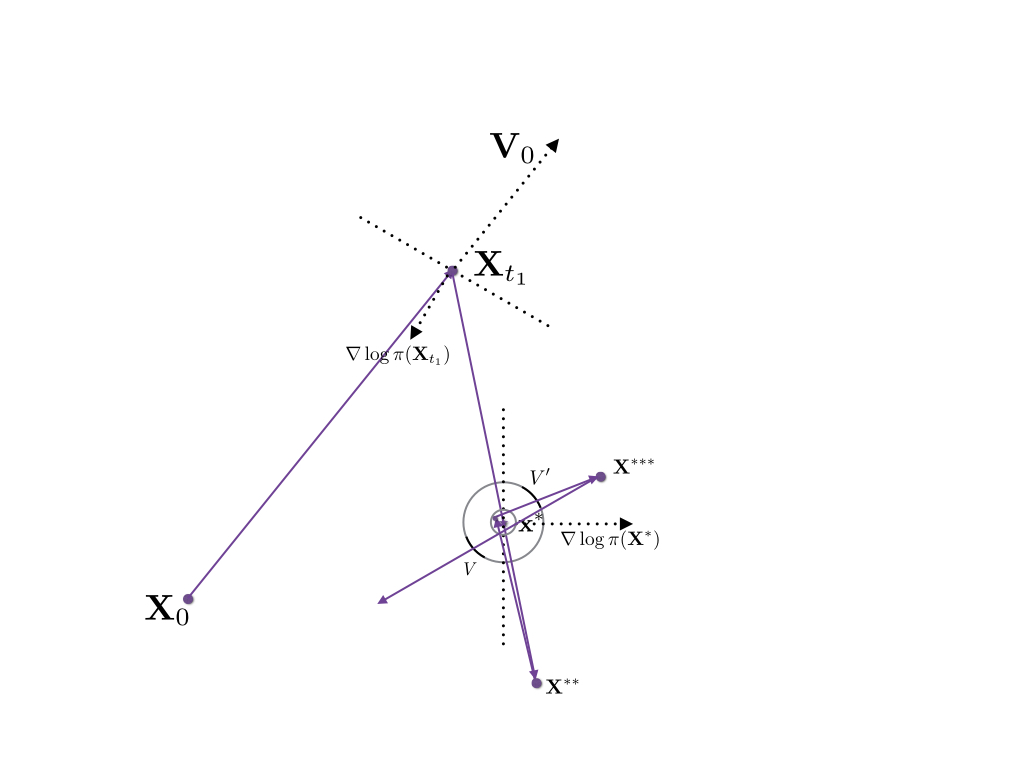}
\caption{Case 3.2}
\label{Case3.2}
\end{figure}
\end{proof}
\noindent With the help of Lemma 1 and Lemma 2, we can prove the ergodicity of GBPS and the way of proof is similar to that of Theorem 1 of \cite{bouchard2017bouncy}
\begin{proof}[Proof of Theorem 2]
Suppose GBPS is not ergodic, then according to Theorem 7.1 of \cite{hairer2010convergence}, there exist two measures $\mu_1$ and $\mu_2$ such that $\mu_1\perp\mu_2$ and $\mu_1$ and $\mu_2$ both are the invariant distribution of GBPS. Thus, there is a measurable set $A\subset \mathbb{R}^d\times S_{d-1}$ such that
\begin{equation*}
\mu_1(A) = \mu_2(A^c) = 0
\end{equation*}
Let $A_1 = A, A_2 = A^c$. By Lemma 2 and Lemma 2.2 of \cite{hairer2010convergence}, the support of $\mu_i$ is $\mathbb{R}^d\times S_{d-1}$. For any open set $B\subset \mathbb{R}^d\times S_{d-1}$, $\mu_i(B) >0$. As a result, at least one of $A_1\cap B$ or $A_2\cap B$ has a positive volume by the measure on $\mathbb{R}^d\times S_{d-1}$ which is induced by Lebesgue measure on $\mathbb{R}^d\times\mathbb{R}^d$. Hence, we denote by $i^*\in\{1, 2\}$ an index satisfying $\text{Vol}(A_{i^*}\cap B) > 0$. From Lemma 2, we know there exists a $z$ in the interior in  $A_{i*}\cap B$ and $t>0$ such that $P_t(z, A_{i^*}) > 0$. As a result, there exist $r_0>0$ and $\delta >0$ such that $P_t(z', A_{i^*}) > \delta$ for all $z'\in B(z, r_0)\subset A_{i*}\cap B$. Hence, 
\begin{equation*}
\begin{split}
\mu_{i^*}(A_{i^*}) &= \int \mu_{i^*}(dz'')P_t(z'', A_{i^*})\\
&\geq\int_B \mu_{i^*}(dz'')P_t(z'', A_{i^*})\\
&\geq \int_{B\cap A_{i^*}} \mu_{i^*}(dz'')\delta\\
&\geq\delta \mu_{i^*}(B\cap A_{i^*}) > 0
\end{split}
\end{equation*}
This contradicts that $\mu_i(A_i) = 0$ for $i\in\{1, 2\}$. As a result, GBPS has as most one invariant measure. By Lemma 1, we complete the proof.
\end{proof}
\newpage

\end{document}